\documentclass{article}
\usepackage[a4paper,twoside=false,footskip=1cm]{geometry}
\usepackage{hyperref}
\usepackage{enumitem}
\usepackage[intlimits]{amsmath}
\usepackage{amssymb,amsthm}
\usepackage{enumitem}
\usepackage{manfnt}
\usepackage{cleveref}
\usepackage[utf8]{inputenc}
\usepackage{pifont}
\usepackage{adjustbox}
\usepackage{dsfont}
\usepackage{subcaption} 

\hypersetup{
linktoc=page,
    unicode=false,          
    pdftoolbar=true,        
    pdfmenubar=true,        
    pdffitwindow=false,     
    pdfstartview={FitH},    
    pdftitle={Consistency of option prices under bid-ask spreads},    
    pdfauthor={I.C. G\"ul\"um, S. Gerhold},     
    pdfsubject={Work in Progress},   
    pdfcreator={ICG},   
    pdfproducer={ICG}, 
    pdfkeywords={Transaction Costs} {Option Prices} 
		{No Arbitrage} {Strassen's Theorem} {Infinity Wasserstein distance}  {Peacock} {Kellerer's Theorem} {Convex Order}, 
    pdfnewwindow=true,      
    colorlinks=true,       
    linkcolor=blue,          
    citecolor=blue,        
    filecolor=blue,      
    urlcolor=blue           
}

\usepackage[english]{babel}

\theoremstyle{plain}
\newtheorem{thm}{Theorem}[section]
\newtheorem{lem}[thm]{Lemma}
\newtheorem{prop}[thm]{Proposition}
\newtheorem{cor}[thm]{Corollary}

\newtheorem{con}[thm]{Conjecture}

\theoremstyle{definition}
\newtheorem{definition}[thm]{Definition}

\theoremstyle{remark}
\newtheorem{exa}[thm]{Example}
\newtheorem{rem}[thm]{Remark}

\theoremstyle{plain}
\newtheorem*{thm*}{Theorem}
\newtheorem*{lem*}{Lemma}
\newtheorem*{prop*}{Proposition}
\newtheorem*{cor*}{Corollary}
\newtheorem*{ass*}{Assumptions}

\theoremstyle{definition}
\newtheorem*{definition*}{Definition}

\theoremstyle{remark}
\newtheorem*{rem*}{Remark}
\newtheorem*{prob*}{Problem}
\newtheorem*{exa*}{Example}

\newcommand{\field}[1]{\mathbb{#1}}
\newcommand{\setN}{\field{N}}

\newcommand{\setR}{\field{R}}

\newcommand{\setM}{\mathcal{M}}

\newcommand{\pp}{{\mathbb P}}

\newcommand{\ind}{\mathds{1}}
\newcommand{\bidy}{\underline{r}}
\newcommand{\asky}{\overline{r}}

\newcommand{\bids}{\underline{S}}
\newcommand{\asks}{\overline{S}}

\newcommand{\leqc}{\leq_{\mathrm{c}}}

\newcommand{\conv}{\mathrm{conv}}
\DeclareMathOperator{\sgn}{sgn}

\setlist[enumerate,1]{label=$(\roman*)$, ref=$(\roman*)$}
\setlist[enumerate,2]{label=\alph*), ref=\theenumi\ \alph*}

\usepackage{csquotes}

\begin{document}

\title{Consistency of option prices under bid-ask spreads}
\author{Stefan Gerhold \\ TU Wien \\ sgerhold@fam.tuwien.ac.at
  \and I. Cetin G\"ul\"um \thanks{We acknowledge financial support from the Austrian Science Fund (FWF) under grant P~24880. We thank the anonymous referees, as well as
  seminar and conference participants at Berlin, the
  12th German Probability and Statistics Days (Bochum), Le Mans, Ulm, Vienna, Oberwolfach, and the 9th BFS Congress (NYC)
  for helpful questions and comments.}
   \\ TU Wien \\ ismail.cetin.gueluem@gmx.net
}
\date{\today}

\maketitle

\begin{abstract}
Given a finite set of European call option prices on a single underlying, we want to know when there is a market model which is consistent with these prices.
In contrast to previous studies, we allow models where the underlying trades at a bid-ask spread.
The main question then is how large (in terms of a deterministic bound)
this spread must be to explain the given prices. We fully solve this problem in the case
of a single maturity, and give several partial results for multiple maturities.
For the latter, our main mathematical tool is a recent result on approximation by
peacocks [S.~Gerhold, I.C.~G\"ul\"um, arXiv:1512.06640].
\end{abstract}

\medskip

\section{Introduction}

\setcounter{equation}{0}
\numberwithin{equation}{section}
Calibrating martingales to given option prices is a central topic of mathematical finance,
and it is thus a natural question which sets of option prices admit such a fit, and which do not.
Note that we are not interested in \emph{approximate} model calibration, but in the consistency
of option prices, 
meaning arbitrage-free models that fit the given prices \emph{exactly}.
Put differently, we want to detect arbitrage in given prices.
We do not consider continuous call price surfaces, but restrict to the
(practically more relevant) case of finitely many strikes and maturities.
Therefore, consider a financial asset with finitely many European call options written on it.
In a frictionless setting, the consistency problem is well understood:
Carr and Madan~\cite{CaMa05} assume that interest rates, dividends and bid-ask spreads are zero, 
and derive necessary and sufficient conditions for the existence of arbitrage free models.
Essentially, the given call prices must not admit calendar or butterfly arbitrage.
Davis and Hobson~\cite{DaHo07} include interest rates and dividends and give similar results. 
They also describe explicit arbitrage strategies, whenever arbitrage exists.
Concurrent related work has been done by Buehler~\cite{Bu06}.
Going beyond existence, Carr and Cousot~\cite{CaCo12} present practically appealing explicit constructions
of calibrated martingales.
More recently,
Tavin~\cite{Ta15} considers options on multiple assets and studies the existence of arbitrage strategies in this setting. Spoida~\cite{Sp14} gives conditions for the consistency of a set of prices
that contains not only vanillas, but also digital barrier options. See~\cite{HeObSpTo16}
for many related references.

As with virtually any result in mathematical finance, robustness with respect to market frictions
is an important issue in assessing the practical appeal of these findings. Somewhat surprisingly,
not much seems to be known about the consistency problem in this direction, the single exception being a paper
by Cousot~\cite{Co07}. He allows positive bid-ask spreads on the options, but not on the underlying,
and finds conditions on the prices that determine the existence of an arbitrage-free
model explaining them.

The novelty of our paper is that we allow a bid-ask spread on the underlying.
Without any further
assumptions on the size of this spread, it turns out that there is no connection between
the quoted price of the underlying and those of the calls: Any strategy trying to
exploit unreasonable prices can be made impossible by a sufficiently large bid-ask spread
on the underlying (see Example~\ref{ex:motivate} and Proposition~\ref{prop:nobound}).
In this respect, the problem
is \emph{not} robust w.r.t.\ the introduction of a spread on the underlying. However,
an arbitrarily large spread seems questionable, given that spreads are usually tight
for liquid underlyings. We thus enunciate
that the appropriate question is not ``when are the given prices consistent'',
but rather ``how large a bid-ask spread on the underlying is needed to explain them?''
Therefore, we put a bound $\epsilon\geq0$ on the spread of the discounted prices, and want to determine the
values of $\epsilon$ that lead to a model explaining the given prices.
We then refer to the call prices as $\epsilon$-consistent (with the absence of arbitrage).
To define the payoff of the call options, we use an arbitrary reference price
process that evolves within the bid-ask spread. We show (Proposition~\ref{prop:ar})
that the consistency problem does not change dramatically if this reference process
is the arithmetic average of the bid and ask prices of the underlying.

Recall that the main technical tool used in the papers~\cite{CaMa05,Co07,DaHo07} mentioned above
to construct arbitrage-free models
is Strassen's theorem~\cite{St65}, or modifications thereof.
In the financial context, this theorem shows the existence of martingale models for
option prices that increase with maturity.
The latter property breaks down if a spread on the underlying is allowed.
We will therefore employ some results from our recent companion paper~\cite{GeGu18},
which deals with variants of Strassen's theorem and approximating sequences
of measures by peacocks (processes increasing w.r.t.\ the convex order).

We assume discrete trading times and finite probability spaces throughout;
no gain in tractability or realism is to be expected by not doing so.
In the case of a single maturity, we obtain simple explicit conditions
that are equivalent to $\epsilon$-consistency (Theorem~\ref{thm:single}).
The multi-period problem, on the other hand, seems to be
challenging. We provide two partial results: necessary (but presumably not sufficient) explicit
conditions for $\epsilon$-consistency (Theorem~\ref{thm:NACVB}), and sufficient semi-explicit conditions (Theorem~\ref{thm:main}). Here, by ``semi-explicit'' we mean the following: Our consistency definition
requires the existence of two sequences of measures, which are not ``too far apart'',
and one of which is a peacock. They correspond to a consistent price system resp.\ to a reference price
that defines the option payoffs. Our result does not say anything about the existence
of the reference price process, but contains explicit conditions for the existence of the peacock.

The structure of the paper is as follows.
In Section~\ref{sec:notation} we  describe our setting and give a precise formulation of our problem.
Also, the significance of peacocks and approximating sequences of measures is explained.
Then, in Section~\ref{sec:sing} we  present
necessary and sufficient conditions for the existence of arbitrage 
free models with bounded bid-ask spreads for a single maturity. 
Our main results on the multi-period problem are contained in Section~\ref{sec:mult eq}.
There, we invoke the main result from~\cite{GeGu18}.
Necessary (but more explicit) conditions for multiple maturities are found in Section~\ref{sec:mult nec}.
Section~\ref{sec:conc} concludes.

\section{The consistency problem under bid-ask spreads}\label{sec:notation}

Our time index set will be  $\mathcal{T}=\{0, \dots, T\}$, where $1\leq T\in\mathbb{N}$,
and~$0$ means today.
By a slight abuse of terminology, we will call the integers in $\mathcal{T}$
``maturities'' and not ``indices of maturities''. We write $\mathcal{T}^*=\{1, \dots, T\}$
for the set of positive times in $\mathcal{T}$.
Whenever we talk about ``the given prices'' or similarly, 
we mean the following data:
\begin{align}
  &\text{A positive deterministic bank account}\
     (B(t))_{t \in \mathcal{T}}\ \text{with}\ B(0)=1, \label{eq:data1} \\
  &\text{strikes}\quad 0< K_{t,1} < K_{t,2}< \dots <K_{t,N_t}, \quad N_t\geq1,\ t\in\mathcal{T}^*, \label{eq:data2}\\
  & \text{corresponding call option bid and ask prices (at time zero)} \notag \\
  & \qquad 0<\bidy_{t,i}  \quad \text{resp.} \quad 0<\asky_{t,i}, \quad 1\leq i\leq N_t,\ t\in\mathcal{T}^*, \label{eq:data3}\\
  & \text{and the current bid and ask price of the underlying} \quad 0<\bids_0 \leq \asks_0.
  \label{eq:data4}
\end{align}
We write $D(t)=B(t)^{-1}$ for the time zero price of a zero-coupon bond maturing at~$t$,
and $k_{t,i}=D(t)K_{t,i}$ for the discounted strikes. The symbol $C_t(K)$ denotes
a call option with maturity~$t$ and strike~$K$.

In the presence of a bid-ask spread on the underlying,
it is not obvious how to define the payoff of an option; this issue seems to have been
somewhat neglected in the transaction costs literature. Indeed, suppose that an agent
holds a call option with strike $\$100$, and that at maturity $T=1$ bid and ask are $\bids_1=\$99$ resp.\
$\asks_1=\$101$. Then, the agent might wish to exercise the option to obtain a security 
for $\$99$ instead of $\$100$,
or he may forfeit the option on the grounds that spending $\$100$
would earn him a position whose liquidation value is only $\$99$. The exercise decision
cannot be nailed down without making further assumptions.
In practice, the quoted ticker price of the underlying is the last price at which an actual transaction
has occurred. This price then triggers cash-settled options. However, this approach
is not feasible in our setup, which does not include an order book.

In the literature on option pricing under transaction costs,
it is usually assumed that the bid and ask of the underlying are {constant multiples
of a mid-price (often assumed to be geometric Brownian motion). This mid-price
is then used as trigger to decide whether an option should be exercised, followed by
physical delivery \cite{Bi14,DaPaZa93,WhWi97}. The assumption that such a constant-proportion
mid-price triggers exercise
seems to be rather ad-hoc, though.
To deal with this problem in a parsimonious way, we assume that call options are cash-settled,
using a reference price process~$S^C$. This process
evolves within the bid-ask spread. It is not a traded asset by itself, but just
serves to fix the call option payoff $(S_{t}^C -K)^+$ for strike~$K$ and maturity~$t$.
This payoff is immediately transferred to the bank account without any costs.
\begin{definition}\label{def:model}
  A \emph{model} consists of a finite probability space $(\Omega, \mathcal{F}, \mathbb{P})$ with a discrete
    filtration $(\mathcal{F}_t)_{t\in\mathcal{T}}$ and three adapted stochastic processes
    $\bids$, $\asks$, and $S^C$, satisfying\footnote{Equations
    and inequalities among random variables are always understood to hold
    almost surely.}
    \begin{equation} \label{eq:scbound}
      0<\bids_t \leq S_t^C \leq \asks_t, \quad t\in\mathcal{T}^* .
    \end{equation}
\end{definition}
Clearly, $\bids_t$ and $\asks_t$ denote the bid resp.\ ask price
of the underlying at time~$t$.
Note that, in our terminology, the initial bid and ask are part of the given
prices (see~\eqref{eq:data4}), and thus the processes in Definition~\ref{def:model} are indexed
by $\mathcal{T}^*=\{1, \dots, T\}$ and not by $\mathcal{T}=\{0, \dots, T\}$.

As for the reference price process~$S^C$, we
do not insist on a specific definition (such as, e.g., $S^C=\tfrac12(\bids+\asks)$), but allow
\emph{any} adapted process inside the bid-ask spread.
We now give a definition for consistency of option prices, allowing for
(arbitrarily large) bid-ask
spreads on both the underlying and the options.
\begin{definition}\label{def:cons}
    The prices~\eqref{eq:data1}--\eqref{eq:data4} are \emph{consistent with the absence of arbitrage,}
    if there is a model (in the sense of Definition~\ref{def:model}) such that
    \begin{itemize}
      \item $\mathbb{E}[(D(t)S_t^C-k_{t,i})^+] \in [\bidy_{t,i}, \asky_{t,i}], \quad
         1\leq i\leq N_t, \ t\in\mathcal{T}^*$,
      \item 
      There is a process $(S^*)_{t\in\mathcal T}$ such that $\bids_t\leq S_t^* \leq \asks_t$ for $t \in \mathcal{T}$ 
    and such that $(D(t)S_t^*)_{t \in \mathcal{T}}$ is a
    $\mathbb{P}$-martingale\footnote{Note that we do not mention the physical probability
    measure, as it is of no relevance to our study.}
    w.r.t.\ the filtration $(\mathcal{F}_t)_{t\in\mathcal{T}}$.
    The pair $(S^*,\mathbb P)$ is called a consistent price system.
    \end{itemize}
\end{definition}
The process~$S^*$ is also called a shadow price.
According to Kabanov and Stricker~\cite{KaSt01} (see also~\cite{Sc04}),
these requirements yield an arbitrage free model comprising
bid and ask price processes for the underlying and each call option. Indeed, for the call
with maturity~$t$ and strike $K_{t,i}$, one may take $\big(\bidy_{t,i}\mathbf{1}_{\{s=0\}}
+B(s)\mathbb{E}[(D(t)S_t^C-k_{t,i})^+|\mathcal{F}_s] \mathbf{1}_{\{s>0\}}\big)_{s\in\mathcal{T}}$
as bid price process (and similarly for the ask price), and
$\big(B(s)\mathbb{E}[(D(t)S_t^C-k_{t,i})^+|\mathcal{F}_s]\big)_{s\in\mathcal{T}}$ as
the process in the second part of Definition~\ref{def:cons}.
We recommend Section~1 of Schachermayer's recent book~\cite{Sc17} as an accessible
introduction to the FTAP under proportional transaction costs.

As mentioned in the introduction, if consistency is defined according to Definition~\ref{def:cons},
then there is no interplay between the current prices of the underlying and the options,
which seems to make little sense. As an illustration,
the following two-period example shows how frictionless arbitrage strategies may fail in the presence
of a sufficiently large spread;
a general result is given in Proposition~\ref{prop:nobound} below.

\begin{exa}\label{ex:motivate}
   Let $c>0$ be arbitrary. We set $k:=k_{1,1}=k_{2,1}=1$ and assume 
   \[ 
     B(1)=B(2)=1, \quad \bids_0=\asks_0=2, \quad r_1:=\bidy_{1,1}=\asky_{1,1} =c+ 1, \quad r_2:=\bidy_{2,1}=\asky_{2,1}= 1.
   \]
   Thus $C_1(k)$ is ``too expensive'', and without frictions, 
   buying  $C_2(k)-C_1(k)$ would be an arbitrage opportunity (upon selling one unit of stock if $C_1(k)$ expires in the money). 
   In particular, the first condition from Corollary 4.2 in~\cite{DaHo07} and equation (5) in~\cite{Co07} 
   are violated: they both state that $r_1 \leq r_2$ is necessary for the absence of arbitrage strategies.

   But with spreads we can choose $c$ as large as we want and still the above prices would be consistent with no-arbitrage.
   Indeed, we can define a deterministic model as follows:
   \[
    \bids_1=\bids_2=2, \quad \asks_1=2c+2, \quad \asks_2=2, \quad S^C= \frac12 (\bids+\asks).
   \]
   Note that 
   \[
     (S_2^C-k)^+= 1 \quad \text{and} \quad (S_1^C-k)^+= c + 1.
   \]
   This model is free of arbitrage (see Proposition~\ref{prop:nobound} below). 
   In particular, consider the portfolio $C_2(k)-C_1(k)$: the short call $-C_1(k)$ finishes in the money with payoff $-(c+1)$.
   This cannot be compensated by going short in the stock, because its bid price stays at~2. 
   The payoff at time $t=2$ of this strategy, with shorting the stock at time~$t=1$, is
   \[
     (S_2^C-k)^+-(S_1^C-k)^+- (\asks_2-\bids_1)=-c <0.
   \]
\end{exa}
\medskip

Our focus will thus be on a stronger notion of consistency, where the discounted spread
on the underlying is bounded. Hence, our goal becomes to determine how large a spread
is needed to explain given option prices.
\begin{definition}\label{def:eps}
  Let $\epsilon\geq0$.  Then the prices~\eqref{eq:data1}--\eqref{eq:data4}  are
  \emph{$\epsilon$-consistent with the absence of arbitrage,} or simply  \emph{$\epsilon$-consistent,}
  if they are consistent (Definition~\ref{def:cons}) and the following conditions hold,
  \begin{align}
    \asks_t - \bids_t &\leq \epsilon B(t), \quad t\in\mathcal{T}, \label{eq:bound eps} \\
    S_t^C &\geq \epsilon B(t),  \quad t\in\mathcal{T}^*. \label{eq:bd SC} 
  \end{align}
\end{definition}
The bound~\eqref{eq:bd SC} is an additional mild assumption on the reference
price~$S^C$, made for tractability, and makes sense given
the actual size of market prices and spreads (recall that $\bids\leq S^C$).
With the same justification, in our main results on $\epsilon$-consistency we will assume that all
discounted strikes~$k_{t,i}$ are larger than~$\epsilon$.
If $\epsilon=0$ and the bid and ask prices in~\eqref{eq:data3}
and~\eqref{eq:data4} agree, then we recover the frictionless
consistency definition from~\cite{DaHo07}.

As mentioned above, we do not insist on any specific definition of the reference price~$S^C$.
However, it is not hard to show that choosing $S^C=\tfrac12(\bids+\asks)$ yields
almost the same notion of $\epsilon$-consistency.

\begin{prop}\label{prop:ar}
   Let $\epsilon \geq 0$ and assume that we are interested in arbitrage free models where, 
   in addition to the requirements of Definition~\ref{def:eps},
   we have that
   \begin{align} \label{eq:midprice}
      S_t^C= \frac{\bids_t+\asks_t}2, \quad  t \in \mathcal{T}^*.
   \end{align}
   Let us then call the prices~\eqref{eq:data1}-\eqref{eq:data4} \emph{arithmetically
   $\epsilon$-consistent}.
   For $\epsilon \geq 0$, the prices are arithmetically $2\epsilon$-consistent if and only
   if they are $\epsilon$-consistent.
\end{prop}
\begin{proof}
   First, assume that there exists an arithmetically $2\epsilon$-consistent model with corresponding stochastic processes $\bids_t, \asks_t, S_t^C, S_t^*$.
   We define new bid and ask prices $\bids_t':= S_t^C \wedge S_t^*$ and $\asks_t':= S_t^C \vee S_t^*$.
   Then~\eqref{eq:midprice} implies that $\asks_t'-\bids_t' \leq B(t) \epsilon$.
   Therefore, the model consisting of $\bids_t', \asks_t', S_t^C, S_t^*$ is $\epsilon$-consistent.
   Conversely, assume that the given prices are $\epsilon$-consistent. 
   Then there exist processes $S^C$ and $S^*$ 
   on a probability space $(\Omega,\mathcal{F},\pp)$ such that $|S_t^C-S_t^*|\leq B(t)\epsilon$ a.s.
   We then simply set $\bids_t=S_t^C-B(t)\epsilon$ and $\asks_t=S_t^C+B(t)\epsilon$, and have thus constructed an arithmetically $2\epsilon$-consistent model.
\end{proof}
Note that the statement of Proposition~\ref{prop:ar}
does not hold for consistency (instead of $\epsilon$-consistency), nor does it
hold if we replace~\eqref{eq:midprice} with
      \begin{align*} \label{eq:midprice_p}
      S_t^C= p \bids_t+(1-p)\asks_t, \quad  t \in \mathcal{T}^*,
   \end{align*}
   where  $p\in [0,1]$ and $p \neq\frac 12$.

The process $(D(t)S_t^{C})_{t \in \mathcal{T}}$ does not have to be a martingale, 
as~$S^C$ is not traded on the market. 
The option prices give us some information about the marginals of the process~$S^C$, though.
On the other hand, the process $(D(t)S_t^{*})_{t \in \mathcal{T}}$
has to be a martingale, but we have no information about its marginals,
except that $|S^*_t-S_t^C| \leq \epsilon B(t)$.
This implies
\begin{equation} \label{Winfappl}
   W^\infty\Bigl( \mathcal{L}\bigl(D(t)S_t^C\bigr),\mathcal{L}\bigl(D(t)S_t^*\bigr) \Bigr)
    \leq \epsilon, \quad t\in\mathcal{T}^*,
\end{equation}
where $W^\infty$ denotes the infinity Wasserstein distance, and $\mathcal{L}$
the law of a random variable. The distance $W^\infty$ is defined on~$\mathcal{M}$,
the set of probability measures on~$\setR$ with finite mean, by
\begin{equation*}
W^\infty (\mu,\nu) = \inf \left\|X-Y \right\|_\infty, \quad \mu,\nu \in\setM.
\end{equation*} 
The infimum is taken over all probability spaces $(\Omega,\mathcal{F}, \pp)$ 
and random pairs $(X,Y)$ with marginals $(\mu,\nu)$.
See~\cite{GeGu18} for some references on~$W^\infty$.
For $\epsilon\geq0$
and random variables~$X$ and~$Y$, the condition $W^\infty(\mathcal{L}X,\mathcal{L}Y)\leq \epsilon$
is equivalent
to the existence of a probability space with random variables $X'\sim \mathcal{L}X,$
$Y'\sim \mathcal{L}Y$ such that $|X'-Y'|\leq \epsilon$ a.s. (This is another result due
to Strassen, see Proposition~\ref{prop:DudleyStrassen} below.)

\begin{definition}\label{def:conv}
Let $\mu, \nu$ be two measures in $\setM$. Then we say that $\mu$ is smaller in \emph{convex order} than $\nu$, in symbols $\mu \leqc \nu$, if for every convex function $\phi: \setR \rightarrow \setR$ we have that $\int \phi \; d \mu \leq \int \phi \; d \nu$, as long as both integrals are well-defined.
A family of measures $(\mu_t)_{t \in \mathcal{T}^*}$ in $\setM$ is called a \emph{peacock}, if $\mu_s \leqc \mu_t$ for all $s \leq t$ in~$\mathcal{T}^*$ (see Definition~1.3 in~\cite{HiPrRoYo11}).
\end{definition}
For $\mu \in \setM$ and $x \in \setR$ we define
\begin{equation}\label{eq:cf}
  R_\mu(x)= \int\nolimits_\setR (y-x)^+  \mu(dy),
\end{equation}
the call function of~$\mu$.
The mean of a measure~$\mu$ will be denoted by $\mathbb{E} \mu = \int y\, \mu(dy)$.
These notions are useful for constructing models for $\epsilon$-consistent prices,
as made explicit by the following lemma. As is evident from its proof,
the sequence~$(\mu_t)$ consists of the marginals of a (discounted) reference price,
whereas~$(\nu_t)$ gives the marginals of a martingale within the bid-ask spread.
The proof uses a coupling result from our companion paper (Lemma~9.1 in~\cite{GeGu18}).
\begin{lem} \label{lem:overview}
   For $\epsilon \geq 0$ the prices \eqref{eq:data1}--\eqref{eq:data4}
   are $\epsilon$-consistent with the absence of arbitrage, 
   if and only if $\asks_0-\bids_0 \leq \epsilon$ and there are sequences of 
   finitely supported measures $(\mu_t)_{t \in \mathcal{T}^*}$
   and $(\nu_t)_{t \in \mathcal{T}^*}$ in $\setM$ such that:
   \begin{enumerate}
     \item $R_{\mu_t}(k_{t,i}) \in [\bidy_{t,i}, \asky_{t,i}]$ for all $t \in \mathcal{T}^*$ and $i \in \{1,\dots, N_t\}$, and $\mu_t([\epsilon,\infty))=1$ for $t\in\mathcal T^*$,
     \item $(\nu_t)_{t \in \mathcal{T}^*}$ is a peacock and its mean satisfies $\mathbb{E}\nu_T \in [\bids_0,\asks_0]$, and
     \item $W^\infty(\mu_t, \nu_t) \leq \epsilon$ for all $t \in \mathcal{T}^*$.
   \end{enumerate}
\end{lem}
\begin{proof}
   Let $(\mu_t)_{t \in \mathcal{T}^*}$ and $(\nu_t)_{t \in \mathcal{T}^*}$ be as above.
   Recall that Strassen's theorem (Theorem~8 in~\cite{St65}) asserts that any peacock
   is the sequence of marginals of a martingale.
   Therefore, there is a finite filtered probability space 
   with a martingale $(\widetilde{S}_t)_{t \in \mathcal{T}}$ such that $\nu_t$ is the law of $\widetilde{S}_t$ for  $t\in\mathcal{T}^*$.

   From~(iii), and the remark before Definition~\ref{def:conv}, it follows  that there is a
   probability space with processes~$\hat M$ and~$\hat{S}^C$ such that $\hat{M}_t \sim \nu_t,$
   $D(t)\hat{S}^C_t \sim \mu_t,$ and $|\hat{M}_t - D(t)\hat{S}^C_t|\leq \epsilon$ for $t\in \mathcal T^*$.
   As in the the proof of Theorem~9.2 in~\cite{GeGu18}, it is easy to see that the finite support 
   condition implies that there is a \emph{finite} probability space with these properties.
   The sufficiency statement now easily follows from Lemma~9.1 in~\cite{GeGu18}. Indeed, that lemma yields
   a finite filtered probability space with adapted processes $(\check S_t)_{t\in\mathcal T}$
   and~$(S_t^C)_{t\in\mathcal T^*}$ satisfying
   \begin{itemize}
     \item $\check S$ is a martingale,
     \item $\check{S}_t\sim \nu_t$ and $D(t)S_t^C \sim \mu_t$ for $t\in \mathcal T^*,$
     \item $|\check{S}_t - D(t) S_t^C| \leq \epsilon$ for $t\in \mathcal T^*.$
   \end{itemize}
   It then suffices to define
   \[ 
     S_t^*:=B(t)\check{S}_t, \quad \bids_t:=S_t^C\wedge S_t^*, \quad \asks_t:=S_t^C\vee S_t^*,
     \quad t\in\mathcal{T}^*,
   \]
   to obtain an arbitrage free model. Note that the second assertion in~(ii) ensures that
   $\bids_t\leq S_t^* \leq \asks_t$ holds for $t \in \mathcal{T}$ and not just $\mathcal T^*$.
   
   Conversely, assume now that the given prices are $\epsilon$-consistent. For $t\in\mathcal T^*$,
   define~$\mu_t$ as the law of $D(t)S_t^C,$ and $\nu_t$ as the law of~$S_t^*$.
   It is then very easy to see that the stated conditions are satisfied. As for the finite support condition,
   note that the probability space in Definition~\ref{def:model} is finite.
\end{proof}
To prepare for the central notions of model-independent and weak arbitrage,
we now define semi-static trading strategies in the bank account, the underlying asset,
and the call options. Here, semi-static means that the position in the call
options is fixed at time zero. The definition is model-independent; as soon
as a model (in the sense of Definition~\ref{def:model}) is chosen, the number of
risky shares in the $t$-the trading period, e.g., becomes
\begin{equation}\label{eq:shares}
  \phi_t^1\big((\bids_u)_{1\leq u< t} , (S_u^C)_{1\leq u< t} , (\asks_u)_{1\leq u< t}\big),
  \quad t\in\mathcal T^*.
\end{equation}
\begin{definition}\label{def:semistatic}
  \begin{enumerate}
   \item
  A \emph{semi-static portfolio}, or \emph{semi-static trading strategy}, is a triple
  \[
    \Phi=\Bigl((\phi^0_t)_{t \in \mathcal{T}^*}, (\phi^1_t)_{t \in \mathcal{T}^*}, (\phi^{t,i})_{t \in \mathcal{T}^*,\, i \in \{1, \dots, N_t\}}\Bigr),
  \]
  where $\phi^0_1\in\mathbb{R}$, $\phi^0_t:(0,\infty)^{3t}\to\mathbb{R}$ are Borel measurable
  for $t\in\mathcal T^*$, analogously for $\phi^1$,
  and $\phi^{t,i}\in\mathbb{R}$ for $t \in \mathcal{T}^*, i \in \{1, \dots, N_t\}$.
  Here, $\phi^0_t$ denotes the investment in the bank account, 
  $\phi^1_t$ denotes the number of stocks held in the period from $t-1$ to~$t$,
  and $\phi^{t,i}\in \setR$ is the number of options with maturity $t \in \mathcal{T}^*$ and strike $K_{t,i}$ which the investor buys at time zero.  
  \item A semi-static portfolio is called \emph{self-financing,} if
  \begin{multline*}
    \phi_{t+1}^0(\boldsymbol{s}_{t}) = \frac{B(t+1)}{B(t)} \phi_{t}^0(\boldsymbol{s}_{t-1})
    + \sum_{i=1}^{N_{t}} \phi^{{t},i}(s_{t}^C-K_{{t},i})^+ \\
    -\big(\phi^1_{t+1}(\boldsymbol{s}_{t})-\phi^1_{t}(\boldsymbol{s}_{t-1})\big)^+ \overline{s}_{t}
    +\big(\phi^1_{t+1}(\boldsymbol{s}_{t})-\phi^1_{t}(\boldsymbol{s}_{t-1})\big)^- \underline{s}_{t}
  \end{multline*}
  holds for $1\leq t< T$ and $\underline{s}_u,s_u^C,\overline{s}_u\in(0,\infty)$,
  $1\leq u\leq t$, where
  \begin{equation}\label{eq:bold s}
    \boldsymbol{s}_t := \big( (\underline{s}_u)_{1\leq u\leq t},(s^C_u)_{1\leq u\leq t},
    (\overline{s}_u)_{1\leq u\leq t} \big).
  \end{equation}
  \item For prices~\eqref{eq:data1}--\eqref{eq:data4}, the \emph{initial portfolio value}
  of a semi-static portfolio~$\Phi$ is given by
  \[
    r_{\Phi}:=\phi_{1}^0 + (\phi_1^{1})^+\asks_0- (\phi_1^{1})^-\bids_0 + \sum_{t\in\mathcal T^*}
    \sum_{i=1}^{N_t} \bigl( (\phi^{t,i})^+\asky_{t,i}- (\phi^{t,i})^-\bidy_{t,i} \bigr).
   \]
   This is the cost of setting up the portfolio~$\Phi$.
   \item The \emph{liquidation value} at time~$T$ is defined as
    \begin{equation*}
      L_\Phi(\boldsymbol{s}_{T}) := \frac{B(T)}{B(T-1)}\phi^0_{T}(\boldsymbol{s}_{T-1})
        + \sum_{i=1}^{N_{T}} \phi^{{T},i}(s_{T}^C-K_{{T},i})^+ 
        -\big(\phi^1_{T}(\boldsymbol{s}_{T-1})\big)^- \overline{s}_{T}
    +\big(\phi^1_{T}(\boldsymbol{s}_{T-1})\big)^+ \underline{s}_{T}.
    \end{equation*}
  \end{enumerate}  
\end{definition}
Having defined semi-static portfolios, we can now formulate two useful notions of arbitrage. 
\begin{definition}\label{def:mNA}
  Let $\epsilon \geq 0$.
    The prices \eqref{eq:data1}--\eqref{eq:data4} admit
    \emph{model-independent arbitrage with respect to spread-bound~$\epsilon$}, if we can form a 
    self-financing semi-static portfolio~$\Phi$
    in the bank account, the underlying asset and the options, such that the initial
    portfolio value~$r_\Phi$ is negative 
    and the following holds: For all real numbers $\underline{s}_t,s_t^C,\overline{s}_t\in(0,\infty)$,
    $1\leq t\leq T$, that satisfy
    \begin{align*}
      0 < \underline{s}_t \leq s_t^C \leq \overline{s}_t,& \quad t \in \mathcal{T}^*, \\
      \overline{s}_t - \underline{s}_t \leq \epsilon B(t),& \quad t \in \mathcal{T}^*, \\
    s_t^C \geq \epsilon B(t), & \quad t \in \mathcal{T}^*,
    \end{align*}
    (cf.~\eqref{eq:scbound}, \eqref{eq:bound eps}, and~\eqref{eq:bd SC}),
    we have $L_\Phi(\boldsymbol{s}_{T})\geq0$.
\end{definition}
\begin{definition}\label{def:wNA}
  Let $\epsilon \geq 0$.
  The prices \eqref{eq:data1}--\eqref{eq:data4} admit
   a \emph{weak arbitrage opportunity with respect to spread-bound~$\epsilon$} 
    if there is no model-independent arbitrage strategy (with respect to spread-bound $\epsilon$),
    but for any model satisfying~\eqref{eq:bound eps} and~\eqref{eq:bd SC}, 
    there is a semi-static portfolio~$\Phi$ such that the initial portfolio value~$r_\Phi$ is non-positive,
    \[
      L_\Phi( (\bids_u)_{1\leq u\leq T} , (S_u^C)_{1\leq u\leq T} , (\asks_u)_{1\leq u\leq T}) \geq0,
    \] and
   \[
      \pp\big(L_\Phi( (\bids_u)_{1\leq u\leq T} , (S_u^C)_{1\leq u\leq T} , (\asks_u)_{1\leq u\leq T})>0\big)>0.
    \]
\end{definition}
Most of the time we will fix $\epsilon \geq 0$ and write only \emph{model-independent arbitrage,} meaning model-independent arbitrage with respect to spread-bound $\epsilon$, and similarly for weak arbitrage.
The notion of weak (i.e., model-dependent) arbitrage was first used in~\cite{DaHo07}, where the authors give examples to highlight the distinction between
weak arbitrage and model-independent arbitrage. The crucial difference is that a weak arbitrage
opportunity may depend on the null sets of the model. E.g., suppose that we would like to use two different
arbitrage strategies according to whether a certain call will expire
in the money with positive probability or not.
Such portfolios could serve to exhibit weak arbitrage (Definition~\ref{def:wNA}),
but will not show model-independent arbitrage (Definition~\ref{def:mNA}).

\section{Single maturity: \texorpdfstring{$\epsilon$}{e}-consistency} \label{sec:sing}

In this section, we characterize $\epsilon$-consistency (according to Definition~\ref{def:eps})
in the special case that all option maturities agree.
The consistency conditions for a single maturity are similar to those derived in Theorem~3.1 of~\cite{DaHo07} and Proposition 3 of~\cite{Co07}. 
In addition to the conditions given there, we have to assume that the mean of $S_1^C$ is ``close enough''  to~$S_0$.

We fix $t=1 \in \mathcal{T}$ and often drop the time index for notational convenience, i.e., we write $\asky_{i}$ instead of $\asky_{1,i}$ etc.
In the frictionless case the underlying can be identified with an option with strike $k=0$. 
Here we will do something similar: 
in the formulation of the next theorem we set $k_0=\epsilon$, as if we would introduce an option with strike 
$\epsilon B(1)$, but we think of $C(\epsilon B(1))$ as the underlying.
The choices for $\bidy_0=\bids_0-2\epsilon$ and $\asky_0=\asks_0$ made in Theorem~\ref{thm:single}
 can be motivated as follows: 
in every model which is $\epsilon$-consistent with the absence of arbitrage, \eqref{eq:bd SC} implies 
that the discounted expected payoff of an option with strike $\epsilon B(1)$ has to satisfy 
\[
   D(1) \mathbb{E}[(S_1^C-\epsilon B(1))^+]= D(1)\mathbb{E}[S_1^C] -\epsilon. 
\]
Furthermore, to guarantee the existence of a consistent price system, $D(1)\mathbb{E}[S_1^C]$ has to lie in the closed interval 
$[\bids_0-\epsilon, \asks_0+\epsilon]$, 
which implies that the price of an option with strike $B(1)\epsilon$ has to lie in the interval $[\bids_0-2\epsilon, \asks_0]$.
Therefore, in the proof of Theorem~\ref{thm:single}
(given in the appendix) we will use the symbol~$C_t(\epsilon B(t))$
as a reference to the underlying 
and $-C_t(\epsilon B(t))$ as a reference to a short position in the underlying plus an additional deposit of~$2\epsilon$ in the bank account.

Before we formulate the main result for a single maturity, we recall that a butterfly contract
(with maturity~$1$)
is defined by
\[
  \frac{1}{K_j-K_i} C_1(K_i)-\bigg(\frac{1}{K_j-K_i}+\frac{1}{K_l-K_j} \bigg) C_1(K_j)
  +\frac{1}{K_l-K_j} C_1(K_l),
\]
where $0\leq i<j<l \leq N,$ and that its payoff is non-negative.
A call spread is a portfolio of a long and a short call, where the latter has
 a larger strike.


\begin{thm}\label{thm:single}
   Let $\epsilon\geq0$ and consider prices as at the beginning of Section~\ref{sec:notation},
   with $T=1$ and $k_1> \epsilon$ (see the remarks after~\eqref{eq:bd SC}).
   Moreover, for ease of notation (see the above remarks)
   we set $k_{0}=\epsilon$, $\bidy_{0}=\bids_0-2\epsilon$,
   and $\asky_{0}=\asks_0$.
   Then the prices are $\epsilon$-consistent (see Definition~\ref{def:eps}) if and only
   if the following conditions hold:
   \begin{enumerate}[label=(\roman*)]
   \item \label{it:NAsing1} All butterfly spreads have non-negative time-0 price, i.e.,
     \begin{equation}\label{eq:NAsing1}
        \frac{\asky_{l}-\bidy_j}{k_{l}-k_j} \geq  \frac{\bidy_j-\asky_{i}}{k_{j}-k_{i}},  \quad  
          0\leq i<j<l \leq N. 
     \end{equation}

   \item \label{it:NAsing2} The call prices satisfy
     \begin{equation}\label{eq:NAsing2}
        \frac{\asky_{l}-\bidy_i}{k_{l}-k_i} \geq -1, \quad  0\leq i<l \leq N.
     \end{equation}

   \item \label{it:NAsing3} All call spreads have  non-negative time-0 price, i.e.,
     \begin{equation}\label{eq:NAsing3}
        \bidy_j \leq \asky_{i}, \quad 0\leq i<j \leq N.
     \end{equation} 

   \item  \label{it:NAsing4} If a call spread is available for zero cost, then the involved options
    have zero bid resp.\ ask price, i.e.,
     \begin{equation}\label{eq:NAsing4} 
        \bidy_j=\asky_{i} \ \Rightarrow \ \bidy_j=\asky_{i}=0,  \quad 0\leq i<j\leq N.
     \end{equation} 
   \end{enumerate}
   Moreover, there is a model-independent arbitrage, as soon as any of the conditions
   \ref{it:NAsing1}--\ref{it:NAsing3} is not satisfied. Finally, if \ref{it:NAsing1}--\ref{it:NAsing3} hold
   but~\ref{it:NAsing4} fails, then there is a weak arbitrage opportunity.
\end{thm}

This theorem is proved in Appendices~\ref{app:proof single} and~\ref{app:proof weak}.
We conclude that the trichotomy of consistency/weak arbitrage/model-independent arbitrage,
which was uncovered by Davis and Hobson~\cite{DaHo07} in the frictionless case,
persists under bid-ask spreads (at least in the one-period setting).
%
%

For $\epsilon=0$ and $\bidy_i=\asky_i=r_i$, the conditions from Theorem~\ref{thm:single} simplify to 
\[
   0 \geq \frac{r_{i+1}-r_i}{k_{i+1}-k_i} \geq  \frac{r_i-r_{i-1}}{k_{i}-k_{i-1}} \geq -1,  \quad \mbox{for} \ i \in \{1,\dots,N-1\},
\] 
and 
\[
   r_i=r_{i-1}  \quad \text{implies} \quad r_i=0, \quad \mbox{for} \ i \in \{1,\dots,N\}.
\]
These are exactly the conditions required in Theorem~3.1 of \cite{DaHo07}.

\begin{rem}
   Note that in contrast to the frictionless case, we do not have to require that bid
   or ask prices decrease as the strike increases,
   in order to get models which are $\epsilon$-consistent with the absence of arbitrage.
   This means that we do not have to require $\bidy_i \geq \bidy_j$ or $\asky_i \geq \asky_j$ for $i < j$,
   as shown in the following example.

   Consider two call options, where $\epsilon=0$ (no spread on the underlying), and the prices are given
   by $\bids_0=\asks_0=5, \ \asky_i=i+5, \ \bidy_i=1+\frac i2, k_i=i$ for $i=1,2$.
   We assume that the bank account is constant until maturity. 
   These prices and a possible choice of 
   shadow prices $e_i:= D(1)\mathbb E[(S^C_1-K_i)^+]$ are shown in Figure~\ref{fig:wachsend}.
   (Note that shadow prices are introduced 
   in the proof of Theorem~\ref{thm:single} in Appendix~\ref{app:proof single}.)
   %
   \begin{figure}[ht]
    \centering
     \includegraphics[width=0.7\textwidth]{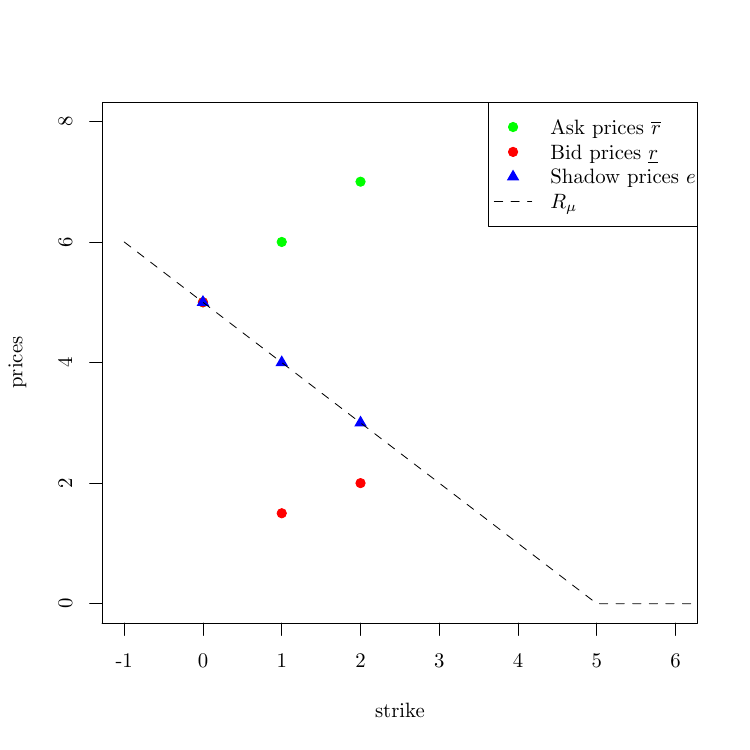}
    \caption{This example shows that it is not necessary that the ask-prices resp.\ bid-prices decrease w.r.t.\ strike. 
   The line represents the call function of $\delta_5$.}
    \label{fig:wachsend} 
   \end{figure}
   Clearly all conditions from Theorem~\ref{thm:single} are satisfied,
   and therefore there exists an arbitrage free model.
   For example we can choose $\mu=\delta_5$, where $\delta$ denotes the Dirac delta.
   This example shows that, in our setting, prices which are admissible from a no-arbitrage point of view
   do not necessarily make economic sense: 
   As the payoff of $C(K_2)$ at maturity never exceeds the payoff of $C(K_1)$, 
   the utility indifference ask-price of $C(K_2)$ should not be higher than the utility indifference  ask-price of $C(K_1)$.
\end{rem}
From Theorem~\ref{thm:single}, it is easy to explicitly compute the interval of
all~$\epsilon$ such that the given prices are $\epsilon$-consistent, which completes
the solution of the $\epsilon$-consistency problem in the one-period case.
Note that \eqref{eq:NAsing1}--\eqref{eq:NAsing4} clearly have to be satisfied
for $i,j,l > 0$, as these conditions depend on~$\epsilon$ only for $i=0$
(see also Proposition~\ref{prop:nobound} below).
\begin{cor} \label{cor:eps_single}
   Assume that the given prices satisfy equations~\eqref{eq:NAsing1}--\eqref{eq:NAsing4}
   for $i,j,l > 0$. Then for $\epsilon \geq 0$ the prices are $\epsilon$-consistent with the absence of arbitrage 
   if and only if $\epsilon$ satisfies:
	\begin{align*}
	   \epsilon &\geq \max\biggl\{ \asks_0 - \bids_0, \bids_0 - \bigl(\asky_i - k_i\bigr), 
		k_j - \frac{\bidy_j - \asks_0}{\asky_l - \bidy_j} \cdot \bigl(k_l - k_j \bigr) \biggr\}, \\
		& \qquad\quad
		\quad 1 \leq i \leq N, \ 1 \leq j < l  \leq N \ \text{such that} \ \asky_l > \bidy_j, \\
\epsilon &\leq \min\biggl\{ k_1, k_j - \frac{\bidy_j - \asks_0}{\asky_l - \bidy_j} \cdot \bigl(k_l - k_j \bigr) \biggr\}, 
\quad 1 \leq j < l  \leq N \ \text{such that} \ \asky_l < \bidy_j.
	\end{align*}
\end{cor}

\begin{proof}
  First, the inequalities $\epsilon \geq \asks_0 - \bids_0$ and $\epsilon \leq k_1$ 
  follow from the definition of $\epsilon$-consistency (see~\eqref{eq:bound eps} and~\eqref{eq:bd SC}). 
  The remaining inequalities follow by setting $i = 0$ in~\eqref{eq:NAsing1} and~\eqref{eq:NAsing2}.
\end{proof}

\section{Multiple maturities: equivalent conditions for consistency
and \texorpdfstring{$\epsilon$}{eps}-consistency}\label{sec:mult eq}

As mentioned in the introduction, our main goal is to find the least bound on the underlying's
bid-ask spread that enables us to reproduce given option prices.
The following result clarifies the situation if \emph{no} such bound is imposed
(see also Example~\ref{ex:motivate}). In our wording, we first seek conditions for
consistency (Definition~\ref{def:cons}) and not $\epsilon$-consistency
(Definition~\ref{def:eps}).
Recall the notation used in, and explained before, Theorem~\ref{thm:single}, where $i=0$ is allowed in
\eqref{eq:NAsing1}-\eqref{eq:NAsing4}, inducing a dependence of these conditions
on~$\bids_0$ and~$\asks_0$. In the following proposition, on the other hand, we require
 $i,j,l \geq1$, and therefore the current bid and ask prices
of the underlying are irrelevant when checking consistency of option prices.
Thus, the notion of $\epsilon$-consistency seems to make
more sense than consistency.

\begin{prop} \label{prop:nobound}
   The prices \eqref{eq:data1}--\eqref{eq:data4}
   are consistent with the absence of arbitrage (see Definition~\ref{def:cons})
   if and only if, 
   for all $t \in \mathcal{T}^*$, the conditions \eqref{eq:NAsing1}--\eqref{eq:NAsing4} from Theorem~\ref{thm:single} hold for $i,j,l \in \{1, \dots, N_t\}$.
\end{prop}
\begin{proof}
   By mimicking the proof of the first part of Theorem~\ref{thm:single} for $i,j,l>0$ we see that the conditions are necessary.
   Now fix $t \in \mathcal{T}^*$ and assume that the conditions hold.
   Exactly as in the sufficiency proof of Theorem~\ref{thm:single},  we can construct $e_{t,1},e_{t,2}, \dots, e_{t,N_t}$ such that
   $e_{t,i} \in [\bidy_{t,i}, \asky_{t,i}]$.
   The linear interpolation $L_t$ of the points $(k_{t,i},e_{t,i})_{i\in \{1,\dots,N_t\}}$ can 
   then be extended to a call function of a measure $\mu_t$ (see the final part of the sufficiency proof of Theorem~\ref{thm:single}).
   
   We define random variables $S_t^C$ such that the law of $D(t)S_t^C$ is given by $\mu_t$. 
   Then we have that 
   \[
     D(t)\mathbb{E}[(S_t^C-K_{t,i})^+]=e_{t,i} \in [\bidy_{t,i}, \asky_{t,i}], \quad i \in \{1,\dots N_t\}. 
   \]
   Furthermore, we pick $s \in [\bids_0,\asks_0]$ and set $\nu_t=\delta_s$ 
   (Dirac delta) for all $t \in \mathcal{T}^*$.
   Clearly, $(\nu_t)_{t \in \mathcal{T}^*}$ is a peacock, and we set $S_t^*=B(t)s$, which implies $D(t)S_t^* \sim \nu_t$.
   Finally, we define $\bids_t= S_t^* \wedge S_t^C$ and $\asks_t=S_t^* \vee S_t^C$, and have thus constructed an arbitrage free model.
\end{proof}
To prepare for our main result on $\epsilon$-consistency in the multi-period model,
we now recall the main result of~\cite{GeGu18},  which gives a criterion
for the existence of the peacock~$(\nu_t)$ from Lemma~\ref{lem:overview}.
Recall also the notation $W^\infty,\mathcal{M}$ introduced before Definition~\ref{def:conv}.
According to Proposition~3.2 in~\cite{GeGu18},
for $\epsilon>0$, a measure $\mu \in \setM$, and $m \in [\mathbb{E}\mu-\epsilon,\mathbb{E}\mu+\epsilon]$,
the set
\[
  \{\nu \in \setM : W^\infty(\mu,\nu)\leq \epsilon,\  \mathbb{E}\nu = m \}
\]
has a smallest and a largest element, and their respective call functions can
be expressed explicitly by the call function~$R_\mu$ of~$\mu$ (see~\eqref{eq:cf}) as follows:
\begin{align*}
R^{\min}_\mu(x;m,\epsilon)&=\Bigl(m+R_\mu(x-\epsilon)-\bigl(\mathbb{E}\mu+\epsilon \bigr) \Bigr) \vee R_\mu(x+\epsilon),  \\
R_\mu^{\max}(x;m,\epsilon)& = \conv\Bigl(m+R_\mu(\cdot +\epsilon)-\bigl(\mathbb{E}\mu-\epsilon \bigr) ,R_\mu(\cdot -\epsilon) \Bigr)(x),
\end{align*}
where $\conv$ denotes the convex hull.
The main theorem of~\cite{GeGu18} gives an equivalent condition for the existence of a peacock
within $W^\infty$-distance $\epsilon$ of a given sequence of measures.
\begin{thm}[Theorem~3.5 in~\cite{GeGu18}]\label{thm:peacocks}
   Let $\epsilon>0$ and $(\mu_n)_{n \in \setN}$ be a sequence in $\setM$ such that 
   \[
     I:= \bigcap_{n \in \setN}[\mathbb{E}\mu_n-\epsilon, \mathbb{E}\mu_n+\epsilon] 
   \]
   is not empty. 
   Then there exists a peacock $(\nu_n)_{n \in \setN}$ 
   such that 
   \begin{equation} \label{eq:Winf}
     W^\infty(\mu_n, \nu_n) \leq \epsilon, \quad \mbox{for all} \ n \in \setN,
   \end{equation}
   if and only if for some $m \in I$
   and for all $N \in \setN$,  $x_1, \dots, x_N \in \setR$,  
   we have
\begin{align} \label{minklmaxwinf}
   R_{\mu_1}^{\min}(x_1; m, \epsilon)+ \sum_{n=2}^{N}\Big(R_{\mu_n}(x_n+\epsilon \sigma_n )-R_{\mu_n}(x_{n-1}+\epsilon \sigma_n )\Big) \leq R_{\mu_{N+1}}^{\max}(x_N; m, \epsilon).
\end{align}
Here,  \textnormal{$\sigma_n= \sgn(x_{n-1}-x_n)$}
depends on~$x_{n-1}$ and~$x_n$.
   In this case it is possible to choose $\mathbb{E}\nu_1=\mathbb{E}\nu_2=\dots=m$. 
\end{thm}
We can now give a partial solution to the multi-period $\epsilon$-consistency problem.
The existence of the measures~$\mu_t$ from Lemma~\ref{lem:overview} (the marginals
of~$DS^C$) has to be assumed, but the existence of the peacock $(\nu_t)$ can be
replaced by fairly explicit conditions, using Theorem~\ref{thm:peacocks}.
\begin{thm}\label{thm:main}
   For $\epsilon \geq 0$ the prices \eqref{eq:data1}--\eqref{eq:data4}
   are $\epsilon$-consistent with the absence of arbitrage, 
   if and only if $\asks_0-\bids_0 \leq \epsilon$ and there is a sequence of
   finitely supported measures $(\mu_t)_{t \in \mathcal{T}^*}$ in $\setM$ such that:
   \begin{enumerate}
     \item $R_{\mu_t}(k_{t,i}) \in [\bidy_{t,i}, \asky_{t,i}]$ for all $t \in \mathcal{T}^*$ and $i \in \{1,\dots, N_t\}$, and $\mu_t([\epsilon,\infty))=1$ for $t\in\mathcal T^*$,
     \item There is
     \[
       m \in \bigcap_{t \in \mathcal T^*}[\mathbb{E}\mu_t-\epsilon, \mathbb{E}\mu_t+\epsilon] 
       \cap [\bids_0, \asks_0]
     \]
     such that for all $N \in \{1,\dots,T-1\}$ and $x_1,\dots,x_N\in\mathbb R$
     \[
       R_{\mu_1}^{\min}(x_1; m, \epsilon)+ \sum_{n=2}^{N}\Big(R_{\mu_n}(x_n+\epsilon \sigma_n )-R_{\mu_n}(x_{n-1}+\epsilon \sigma_n )\Big) \leq R_{\mu_{N+1}}^{\max}(x_N; m, \epsilon),
     \]
     where $\sigma_n$ is as in Theorem~\ref{thm:peacocks} and $\mu_n := \mu_T$
     for $n>T$.
   \end{enumerate}
\end{thm}
\begin{proof}
   Immediate from Lemma~\ref{lem:overview} and Theorem~\ref{thm:peacocks}.
\end{proof}
As we allow an arbitrary reference price process~$S^C$ in Definitions~\ref{def:model}
and~\ref{def:cons}, our notion of consistency is fairly weak. It can be weakened further
by requiring that the bound~\eqref{eq:bound eps} holds only with a certain probability
instead of almost surely. However, according to the following theorem,
we can always find such a model as soon as the prices are consistent.
\begin{thm} \label{thm:pbound}
   Let $p \in (0,1]$ and $\epsilon \geq 0$.
   For given prices \eqref{eq:data1}--\eqref{eq:data4} the following are equivalent:
   \begin{itemize}
      \item[(i)] The prices satisfy Definition~\ref{def:eps} ($\epsilon$-consistency), but
      with~\eqref{eq:bound eps} replaced by the weaker condition
      \[
          \mathbb{P} \Bigl( \asks_t-\bids_t \geq \epsilon B(t)\Bigr) \leq p, \quad t \in \mathcal{T}.
      \]
      \item[(ii)]
      The prices are consistent with the absence of arbitrage.
   \end{itemize}
\end{thm}
For the proof of Theorem~\ref{thm:pbound} we employ a result from~\cite{GeGu18}
on the modified Prokhorov distance.
\begin{definition}
For $p \in [0,1]$  and two probability measures $\mu,\nu$ on $\setR$,
 we define the \emph{modified Prokhorov distance} as
   \begin{align*}
     d_p^\mathrm{P}(\mu,\nu):= \inf \Bigl\{h >0: \nu(A) \leq \mu(A^h)+p, \ \text{for all closed sets} \ A \subseteq \setR \Bigr\}.
   \end{align*}
\end{definition}
(To define the standard Prokhorov distance, replace~$p$ by~$h$ in the right-hand side.)
Note that $d_0^\mathrm{P}=W^\infty.$
A well known result, which was first proved by Strassen, and was then extended by Dudley~\cite{Du68,St65}, explains the connection of $d_p^\mathrm{P}$ to minimal distance couplings.
\begin{prop} \label{prop:DudleyStrassen}
   Given measures $\mu, \nu$ on $\setR$, $p \in [0,1]$, and $\epsilon>0$, there exists a probability space~$(\Omega, \mathcal{F}, \pp)$ with random variables $X\sim \mu$ and $Y \sim \nu$ such that 
   \begin{equation} \label{smalldist} 
     \pp\big( \bigl|X-Y|> \epsilon  \big) \leq p,
   \end{equation}
   if and only if
    \begin{equation} 
     d^\mathrm{P}_p(\mu,\nu) \leq \epsilon.
   \end{equation}
\end{prop}
The following result shows that, unlike for~$W^\infty$, there \emph{always} exists
an approximating peacock w.r.t.\ $d_p^\mathrm{P}$ for $0<p\leq1$. This explains
why the very weak condition of consistency is sufficent to imply~(i)
in Theorem~\ref{thm:pbound}.
\begin{thm}[Theorem~8.3 in~\cite{GeGu18}] \label{thm: dppstrassen}
   Let $(\mu_n)_{n \in \setN}$ be a sequence in $\setM$, $\epsilon>0$, and $p \in (0,1]$. Then, for all $m \in \setR$ there exists a peacock $(\nu_n)_{n \in \setN}$ with mean $m$ such that
   \[
     d_p^\mathrm{P}(\mu_n, \nu_n) \leq \epsilon.
   \]
\end{thm}
\begin{proof}[Proof of Theorem~\ref{thm:pbound}]
   (i) implies~(ii) by definition.
   To show the other implication, 
   we define probability measures $(\mu_t)_{t \in \mathcal{T}^*}$ as in 
   the proof of Proposition~\ref{prop:nobound},
   such that $R_{\mu_t}(k_{t,i}) \in [\bidy_{t,i}, \asky_{t,i}]$ for $i \in \{1,\dots N_t\}$
   and $t \in \mathcal{T}^*$.
   Now we pick $s \in [\bids_0,\asks_0]$. Then by, Theorem~\ref{thm: dppstrassen},
   there exists a peacock $(\nu_t)_{t \in \mathcal{T}^*}$ with mean $s$
   such that $d_p^\mathrm{P}(\mu_t,\nu_t)\leq \epsilon$ for all $t \in \mathcal{T}^*$. 
   We can now use Proposition~\ref{prop:DudleyStrassen} and proceed as in the proof
   of Lemma~\ref{lem:overview} to conclude that there exist stochastic processes 
   $(\widetilde{S}^C_t)_{t \in \mathcal{T}^*}$ and $(\widetilde{S}^*_t)_{t \in \mathcal{T}^*}$  whose marginal distributions 
   are given by $(\mu_t)_{t \in \mathcal{T}^*}$ resp.\ $(\nu_t)_{t \in \mathcal{T}^*}$, such that $(\widetilde{S}^*_t)_{t \in \mathcal{T}^*}$ is a martingale and such that
   \begin{align*} 
       \mathbb{P} \Bigl( \bigl|\widetilde{S}^*_t-\widetilde{S}^C_t \bigr| \geq \epsilon \Bigr) \leq p, \quad t \in \mathcal{T}^*.
    \end{align*}
   The coupling lemma we use (Lemma~9.1 in~\cite{GeGu18}) was formulated in~\cite{GeGu18}
   for the special case $p=0$, but the proof trivially extends to $p\in[0,1]$.
   We then simply put 
   \[
S^*_t=B(t)\widetilde{S}^*_t, \quad S_t^C=B(t)\widetilde{S}^C_t, \quad \bids_t= S_t^* \wedge S_t^C, \quad \mbox{and} \ \asks_t=S_t^* \vee S_t^C. \qedhere
   \]
\end{proof}

\section{Multiple maturities: necessary conditions for \texorpdfstring{$\epsilon$}{eps}-consistency} \label{sec:mult nec}

The main result of the preceding section (Theorem~\ref{thm:main}) gives semi-explicit
equivalent conditions for $\epsilon$-consistency. The goal of the present section
is to provide explicit necessary conditions.
For a single maturity, the $\epsilon$-consistency conditions (Theorem~\ref{thm:single})
are a generalization of the frictionless conditions in~\cite{Co07,DaHo07}.
They guarantee that for each  maturity $t \in \mathcal{T}^*$ the option prices 
can be associated to a measure $\mu_t$, such that $\mathbb{E}\mu_t \in [\bids_0,\asks_0]$
(cf.\ Lemma~\ref{lem:overview}).
In this section we state \emph{necessary} conditions for multiple periods.
Our conditions (see Definition~\ref{def:CVB} and Theorem~\ref{thm:NACVB})
are fairly involved, and we thus expect that it might not be
easy to obtain tractable \emph{equivalent} conditions. 
In the case where there is only a spread on the options, but not
on the underlying, it suffices to compare prices with only three or two different maturities
(see equations (4), (5) and (6) in~\cite{Co07} and Corollary 4.2 in~\cite{DaHo07})
to obtain suitable consistency conditions. These conditions ensure that the family of measures $(\mu_t)_{t \in \mathcal{T}^*}$ is a peacock.

If we consider a bid-ask spread on the underlying and
want to check for $\epsilon$-consistency according to
Definition~\ref{def:eps} ($\epsilon > 0$), it turns out that
we need conditions that involve all maturities simultaneously (this will become clear
by condition~\eqref{minklmaxwinf} below).
We thus introduce calendar vertical baskets (CVB), portfolios which consist of various long and short positions in the call options.
We first give a definition of CVBs.
Then, in Lemma~\ref{le:CVB} we will study a certain trading strategy involving a short position in a CVB.
This strategy will then serve as a base for the conditions in Theorem~\ref{thm:NACVB},
which is the main result of this section.
Note that our definition of a CVB depends on $\epsilon\geq 0$:
the contract defined in Definition~\ref{def:CVB} only provides necessary conditions in markets 
where the bid-ask spread is bounded by $\epsilon\geq 0$.

\begin{definition} \label{def:CVB}
   Fix $u \in \mathcal \{1, \dots, T-1\}$ and $\epsilon\geq 0$ and assume that vectors 
   $\boldsymbol{\sigma}=(\sigma_1,\dots,\sigma_u)$, $\boldsymbol{x}=(x_1,\dots,x_u)$, $\boldsymbol{I}=(i_2, \dots, i_{u})$ and $\boldsymbol{J}=(j_1, \dots, j_u)$
   are given, such that 
   \begin{enumerate} [label=(\textit{\roman*})]
     \item $x_t \in \setR$ for all $t \in \{1,\dots, u\}$,
     \item $\sigma_1 \in \{-1,1\}$ and $\sigma_t=\sgn(x_{t-1}-x_t)$ for all $t \in \{2,\dots,u\}$,
     \item $j_t \in \{0,\dots, N_{t}\}$ and $k_{t,j_t} = x_t+\epsilon\sigma_{t}$ for all $t \in \{1,\dots, u\}$,
     \item $i_t \in \{0,\dots, N_{t}\}$ and either $k_{t,i_t} \leq x_{t-1}+\epsilon\sigma_{t}$ or $i_t=0$ for all $t \in \{2,\dots, u\}$.
   \end{enumerate}
   Then we define a \emph{calendar vertical basket} with these parameters as the contract
   \begin{align} \label{eq:CVB1}
     CVB_u(\boldsymbol{\sigma},\boldsymbol{x},\boldsymbol{I},\boldsymbol{J}) =
     C_1(K_{1,j_1}) + \sum_{t=2}^u \Bigl(C_t\bigl(K_{t,j_t}\bigr)- C_t\bigl(K_{t,i_{t}}\bigr)\Bigr)-2\epsilon \ind_{\{\sigma_1=-1\}}.
   \end{align}
   The market ask resp.\ bid-price of $CVB_u(\boldsymbol{\sigma},\boldsymbol{x},\boldsymbol{I},\boldsymbol{J})$ are given by
   \begin{align} \nonumber
      \asky^{CVB}_u(\boldsymbol{\sigma},\boldsymbol{x},\boldsymbol{I},\boldsymbol{J}) &= \asky_{1,j_1} + \sum_{t=2}^u \bigl(\asky_{t,j_t}-\bidy_{t,i_{t}} \Bigr) -2\epsilon \ind_{\{\sigma_1=-1\}}, \\
      \bidy^{CVB}_u(\boldsymbol{\sigma},\boldsymbol{x},\boldsymbol{I},\boldsymbol{J}) &= \bidy_{1,j_1} + \sum_{t=2}^u \bigl(\bidy_{t,j_t}-\asky_{t,i_{t}} \Bigr) +2\epsilon \ind_{\{\sigma_1=-1\}} \label{eq:CVBprice}.
   \end{align}
   We will refer to~$u$ as the maturity of the CVB.
\end{definition}

\begin{lem} \label{le:CVB}
   Fix $\epsilon\geq 0$.
   For all parameters $u,\boldsymbol{\sigma},\boldsymbol{x},\boldsymbol{I},\boldsymbol{J}$ as in Definition~\ref{def:CVB},
   there is a self-financing semi-static portfolio~$\Phi$ whose initial value is given by $r_\Phi=-\bidy^{CVB}_u(\boldsymbol{\sigma},\boldsymbol{x},\boldsymbol{I},\boldsymbol{J})$,
   such that for all models satisfying~\eqref{eq:bound eps} and~\eqref{eq:bd SC} 
	and for all $t \in \{2,\dots,u+1\}$ 
	one of the following conditions holds:
   \begin{enumerate}
   \item $\phi_{t}^0 \geq 0$ and $\phi_t^1=0$, or  
   \item $\phi_{t}^0 \geq k_{t,j}-\epsilon\sigma_t $ and $\phi_t^1=-1$.
   \end{enumerate}
   In particular, all corresponding cash-flows are non-negative.
\end{lem}
The arguments of $\phi_t^0,\phi_t^1$ are of course the
same as in~\eqref{eq:shares}, and are omitted for brevity.
In the proof of Lemma~\ref{le:CVB}, we define the functions  $\phi_t^0,\phi_t^1$ inductively.
As we are defining a model-independent strategy, we could also use the deterministic dummy
variables~\eqref{eq:bold s}
from Definition~\ref{def:semistatic} as arguments. It seems more natural to write
$(\bids_u)_{u\leq t} , (S_u^C)_{u\leq t} , (\asks_u)_{u\leq t}$,
though. We just have to keep in mind that $\phi_t^0,\phi_t^1$
have to be constructed as \emph{functions} of
$(\bids_u)_{u\leq t} , (S_u^C)_{u\leq t} , (\asks_u)_{u\leq t}$, without using the
\emph{distribution} of these random vectors. \\
Moreover, note that later on in Theorem~\ref{thm:NACVB} we will only need the case where $u < T$, therefore we excluded the case $u=T$.

\begin{proof}[Proof of Lemma~\ref{le:CVB}]   
   Assume that we buy the contract 
   \begin{align} \label{eq:CVB2}
    -CVB_u(\boldsymbol{\sigma},\boldsymbol{x},\boldsymbol{I},\boldsymbol{J})= -C_1(K_{1,j_1}) + \sum_{t=2}^u \Bigl( C_t\bigl(K_{t,i_{t}}\bigr)-C_t\bigl(K_{t,j_t}\bigr)\Bigr)+2\epsilon \ind_{\{\sigma_1=-1\}},
   \end{align}
   thus we are getting an initial payment of $\bidy^{CVB}_u(\boldsymbol{\sigma},\boldsymbol{x},\boldsymbol{I},\boldsymbol{J})$.
   We have to keep in mind that if $i_t=0$ for some $t \in \{2,\dots, u\}$, then the
   corresponding expression in~\eqref{eq:CVB2} denotes a long position in the underlying, and if
   $j_t=0$ for some $t \in \{1,\dots, u\}$, then the expression~$-C_t(K_{t,j_t})$ in~\eqref{eq:CVB2} denotes a short position in the underlying 
   plus an additional deposit of $2\epsilon$ in the bank account at time~0
   (see the beginning of Section~\ref{sec:sing}).
   To ease notation, we will write $K_{t,i}$ instead of $K_{t,i_{t}}$ and $K_{t,j}$ instead of $K_{t,j_t}$.

   We will show inductively that after we have traded at time $t \in \{1,\dots,u\}$ we can end up in one of two scenarios:
   either the investor holds a non-negative amount of bank units (i.e., $\phi_{t+1}^0\geq 0$), we will call this scenario~A, 
   or we have one short position in the underlying (i.e., $\phi^1_{t+1}=-1$) and $\phi^0_{t+1}\geq k_{t,j}-\epsilon\sigma_t $; 
   we will refer to this as scenario~B.
   Note that scenarios~A and~B are not disjoint, but this will not be a problem.

   We will first deal with the case where $\sigma_1=-1$ and afterwards with the case $\sigma_1=1$.
   We start with $t=1$ and first assume that $j_1>0$.
   If $C_1(K_{1,{j}})$ expires out of the money, then we do not trade at time~1 and obtain $\phi^0_2 = 2 \epsilon \geq 0$, 
   so we are in scenario~A.
   Otherwise we sell one unit of the underlying, and thus
   \[
    \phi^0_2=2\epsilon+k_{1,j} +D(1)\bigl(\bids_1-S_1^C\bigr) \geq k_{1,j}+\epsilon =k_{1,{j}}-\sigma_1\epsilon,
   \]
   yielding scenario~B. Recall from Section~\ref{sec:notation} that $D(t)=B(t)^{-1}$.
   If $j_1=0$ then $k_{1,j}=\epsilon$.
   We do not close the short position in this case and 
   we get that $\phi^0_2 = 4\epsilon \geq k_{1,{j}}-\sigma_1\epsilon$, so we also get to scenario~B.

   For the induction step we split the proof into two parts. 
   In part~A we will assume that after trading time~$t-1$ we are in scenario~A,
   and in part~B we will assume that at the end of period~$t-1$ we are in scenario~B.

   \textbf{Part A}: 
   We will show that after we have traded at time~$t$ 
    we can end up  either in situation~A or~B.
   First we assume that $j_{t},i_{t}>0$, and so both expressions in~\eqref{eq:CVB2} with maturity~$t$ denote options (and not the underlying).
   Under these assumptions~$\phi_t^0$ satisfies
   \[
     \phi_{t+1}^0 \geq D(t)\bigl(S^C_{t}-K_{t,i}\bigr)^+-D(t)\bigl(S^C_{t}-K_{t,j}\bigr)^+.
   \]
   Clearly, if $K_{t,i} \leq K_{t,j}$ or if both options expire out of the money, then $\phi_{t+1}^0 \geq 0$,
   and we are in situation~A.
   So suppose that $K_{t,i} > K_{t,j}$ and that $S_t^C>K_{t,j}$. This also implies that $\sigma_{t}=1$.
   If this is the case, we go short one unit of the underlying, and~$\phi_{t+1}^0$ can be bounded from below as follows,
   \begin{align*}
    \phi_{t+1}^0 &\geq D(t)\bigl(S^C_{t}-K_{t,i}\bigr)^+-D(t)\bigl(S^C_{t}-K_{t,j}\bigr)+D(t)\bids_{t} \\
    & \geq k_{t,j}-\epsilon\sigma_{t}.
   \end{align*}
   This corresponds to situation~B. 
   Next assume that $j_t=0$ and $i_t>0$. Then we have that $k_{t,j}=\epsilon$.
   After trading time~$t$ we end up in scenario~B,
   \begin{align*}
     \phi_{t+1}^0 &\geq  D(t)\bigl(S^C_{t}-K_{t,i}\bigr)^+ +2\epsilon \geq k_{t,j}-\epsilon\sigma_{t}.
   \end{align*}
   We proceed with the case that $j_t>0$ and $i_{t}=0$. 
   As $k_{t,j}>\epsilon$, we can close the long position in the underlying and end up in scenario~A at the end of time~$t$,
   \begin{align*}
     \phi_{t+1}^0 \geq D(t)\bids_t - D(t)\bigl(S^C_{t}-K_{t,j}\bigr)^+ \geq 0.
   \end{align*}
   The case where $j_t=i_{t}=0$ is easily handled, because the long and the short position simply cancel out.
   We are done with part~A.

   \textbf{Part B:}
   Assume that after we have traded at time~$t-1$ we are in scenario~B, and thus 
   $
     \phi_{t}^0 =k_{t-1,j}- \epsilon\sigma_{t-1}.
   $
   First we will consider the case where $j_{t},i_{t}>0$.
   If at time~$t$ the option with strike $K_{t,j}$ expires in the money, we do not close the short position and have
   \begin{align*}
     \phi_{t+1}^0 & \geq \phi_{t}^0 + D(t)\bigl(S^C_{t}-K_{t,i}\bigr)^+-D(t)\bigl(S^C_{t}-K_{t,j}\bigr) \\
    &=  k_{t-1,j} -\epsilon\sigma_{t-1} + k_{t,j} - k_{t,i} \\
    &\geq k_{t,j} -\epsilon\sigma_{t},
   \end{align*}
   which means that we end up in scenario~B. 
  Now we distinguish two cases according to $x_{t-1} \leq x_{t}$ and $x_{t-1} > x_{t}$, 
  and always assume that $C_t(K_{t,j})$ expires out of the money.
   If $x_{t-1} \leq x_{t}$, then we also have that $k_{t,i} \leq k_{t,j}$ and that $\sigma_{t}=-1$.
   We close the short position to end up in scenario~A,
   \begin{align*}
     \phi_{t+1}^0 & \geq \phi_{t}^0 + D(t)\bigl(S^C_{t}-K_{t,i}\bigr)^+-D(t)\asks_{t} \\
            & \geq k_{t,i}- \epsilon\sigma_{t} - k_{t,i} -\epsilon \geq 0.
   \end{align*}
   If on the other hand $x_{t-1} > x_{t}$ and $\sigma_{t}=1$, we do not trade at time~$t$ to stay in scenario~B,
   \begin{align*}
     \phi_{t+1}^0
     & \geq \phi_{t}^0 + D(t)\bigl(S^C_{t}-K_{t,i}\bigr)^+ \\
     & > k_{t,j}- \epsilon\sigma_{t}.
   \end{align*}

   We proceed with the case where~$j_t=0$ and~$i_t>0$. 
   As before, we have $k_{t,j}=\epsilon$, and we can close one short position to stay in scenario~B,
   \begin{align*}
     \phi_{t+1}^0 &=\phi_{t}^0  + D(t)\bigl(S^C_{t}-K_{t,i}\bigr)^+ + 2\epsilon-D(t)\asks_{t}\\
     &\geq k_{t-1,j} -\epsilon\sigma_{t-1} -k_{t,i}+ \epsilon\\
     &\geq \epsilon - \epsilon\sigma_{t}= k_{t,j}- \epsilon\sigma_{t}.
   \end{align*}
   If $j_t>0$ and $i_{t}=0$, then we distinguish two cases: 
   either $C_t(K_{t,j})$ expires out of the money, in which case we cancel out the long and short position in the underlying and have
   \[
    \phi_{t+1}^0 \geq \phi_{t}^0 \geq 0,
   \]
   which corresponds to scenario~A. Or, $C_t(K_{t,j})$ expires in the money.
   Then we sell one unit of the underlying and hence we end up in scenario~B,
   \begin{align*}
    \phi_{t+1}^0 &\geq \phi_{t}^0 -D(t)\bigl(S_t^C-K_{t,j}\bigr)+ D(t) \bids_t \\ 
             &\geq k_{t-1,j} -\epsilon\sigma_{t-1} + k_{t,j} - \epsilon \\
             &\geq k_{t,j}-\epsilon\sigma_{t}.
   \end{align*}
   In the last inequality we have used that $k_{t-1,j} -\epsilon\sigma_{t-1} = x_{t-1} \geq k_{t,i}-\epsilon\sigma_{t}$,
   and that $k_{t,i} =\epsilon$.

   The case where  $j_t=i_{t}=0$ is again easy to handle, because the long and the short position cancel out 
   and we are in scenario~B at the end of the $(t+1)$-st period.

   Thus after we have traded at time~$u$ we are either in scenario~A or scenario~B, which proves the assertion if $\sigma_1=-1$.

   The proof for $\sigma_1=1$ is similar.
   We will first show that after trading at time 1 we can either be in scenario~A or scenario~B, and the statement of the proposition then follows 
   by induction exactly as in the case $\sigma_1=-1$.

   First we assume that $j_1>0$. Then, if the option $C_1(K_{1,j})$ expires out of the money, we are in scenario~A;
   otherwise we go short in the underlying and have
   \begin{align*}
    \phi_2^0 & \geq -D(1)\bigl(S_1^C-K_{1,j}\bigr)+ D(1)\bids_1 \geq k_{1,j}-\epsilon,
   \end{align*}
   which corresponds to scenario~B.
   If $j_1=0$, then we also have that $k_{j,1} =\epsilon$, and hence we are in scenario~B.
\end{proof}

According to Lemma~\ref{le:CVB}, there is a semi-static, self-financing trading strategy~$\Phi$ for the buyer of the contract $-CVB_u(\boldsymbol{\sigma},\boldsymbol{x},\boldsymbol{I},\boldsymbol{J})$,
such that~$(\phi_{u+1}^0,\phi^1_{u+1})$ only depends on $\sigma_u, k_{u,j}$ (the investor might have some surplus in the bank account).
In the following we will use this strategy and only write $-CVB_u(\sigma_u, k_{u,j})$ resp.\ $\bidy_u^{CVB}(\sigma_u, k_{u,j})$
instead of $-CVB_u(\boldsymbol{\sigma},\boldsymbol{x},\boldsymbol{I},\boldsymbol{J})$ resp.\ $\bidy_u^{CVB}(\boldsymbol{\sigma},\boldsymbol{x},\boldsymbol{I},\boldsymbol{J})$.
In the case where $\phi_u^0\geq 0$ and $\phi_u^1=0$  we will say that the calendar vertical basket expires out of the money; otherwise we will say that it expires in the money.

The next theorem states necessary conditions for the absence of arbitrage in markets with spread-bound $\epsilon\geq 0$.
\begin{thm} \label{thm:NACVB}
   Let $\epsilon\geq 0$, $s,t,u \in \mathcal{T}$ such that $s < t$ and $s < u$ and $i \in \{0,\dots,N_{t}\}$, $j \in \{0,\dots,N_{s}\}$, $l \in \{0,\dots,N_{u}\}$.
   Fix prices as at the beginning of Section~\ref{sec:notation}, with $k_{t,1}>\epsilon$
   for all $t\in\mathcal{T}$.
   Then, for all calendar vertical baskets with maturity $s \in \mathcal{T}$ and parameters $k_{s,j}$ and $\sigma_{s}$,
   the following conditions are necessary for $\epsilon$-consistency,
   \begin{enumerate}[label=(\textit{\roman*})]
     \item \label{it:NACVB1} 
       \begin{align} \label{eq:NACVB1}
         \frac{\bidy_s^{CVB}(\sigma_s,k_{s,j})-\asky_{t,i}}{\bigl(k_{s,j}-\epsilon\sigma_{s}\bigr)-\bigl(k_{t,i}+\epsilon \bigr)}  \leq   
         \frac{\asky_{u,l}-\bidy_s^{CVB}(\sigma_s,k_{s,j})}{k_{u,l}+\epsilon-\bigl(k_{s,j}-\epsilon\sigma_{s}\bigr)},
           \quad \text{if} \quad k_{t,i}+\epsilon < k_{s,j}- \epsilon\sigma_s  < k_{u,l}+\epsilon,
     \end{align}
     \item \label{it:NACVB2} 
       \begin{align} \label{eq:NACVB2}
         \frac{\asky_{u,l}-\bidy_s^{CVB}(\sigma_s,k_{s,j})}{k_{u,l}+\epsilon-\bigl(k_{s,j}-\epsilon\sigma_{s}\bigr)} \geq -1, \quad \text{if} \quad k_{s,j}- \epsilon\sigma_s  < k_{u,l}+\epsilon,
     \end{align}
     \item \label{it:NACVB3}
     \begin{align} \label{eq:NACVB3}
        \bidy_s^{CVB}(\sigma_s,k_{s,j})-\asky_{t,i} &\leq 0, \quad \text{if} \quad k_{s,j}- \epsilon\sigma_s  \geq k_{t,i}+\epsilon, 
      \end{align}
     \item \label{it:NACVB4} 
       \begin{align} \label{eq:NACVB4}
          \bidy_s^{CVB}(\sigma_s,k_{s,j})-\asky_{t,i}=0 \ \Rightarrow \ \asky_{t,i}=0,   \quad \text{if} \quad k_{s,j}- \epsilon\sigma_s  >k_{t,i}+\epsilon.
     \end{align}
   \end{enumerate}

\end{thm}

\begin{proof}
We will assume that $s < t\leq u$ and that $i,l>0$. The other cases can be dealt with similarly.
In all four cases \ref{it:NACVB1}--\ref{it:NACVB4} we will assume that until time~$s$ we followed the trading strategy described in 
Lemma~\ref{le:CVB}.

   \ref{it:NACVB1}
   If~\eqref{eq:NACVB1} fails, then we set
   \[ 
     \theta= \frac{k_{u,l}+\epsilon-\bigl(k_{s,j}-\epsilon\sigma_s\bigr)}{k_{u,l}-k_{t,i}} \in (0,1)
   \]
   and buy $\theta C_t(K_{t,i})+(1-\theta) C_u(K_{u,l})- CVB_s(\sigma_s, K_{s,j})$, making an initial profit. 
   If the calendar vertical basket $CVB_s(\sigma_s, K_{s,j})$ expires out of the money, then we have model-independent arbitrage. 
   Otherwise we have a short position in the underlying at time~$s$. In order to close the short position, we buy~$\theta$ units of the underlying at time~$t$,
   and we buy $1-\theta$ units of the underlying at time~$u$.
   The liquidation value of this strategy at time~$u$ is then non-negative,
   \begin{align*}
    (k_{s,j}-\epsilon\sigma_{s}+\epsilon)B(u) &+  \theta(S^C_{t}-K_{t,i})^+  \frac{B(u)}{B(t)}+ (1-\theta) (S^C_{u}-K_{u,l})^+ \\
          &+(\bids_s-S_s^C)\frac{B(u)}{B(s)}-\theta \asks_{t}\frac{B(u)}{B(t)}-(1-\theta)\asks_{u} \\[2mm]
   \geq \ & (k_{s,j}-\epsilon\sigma_{s})B(u) + \theta \frac{B(u)}{B(t)} \Bigl(S^C_{t}-K_{t,i}-\asks_{t} \Bigr) + (1-\theta) \Bigl(S^C_{u}-K_{u,l}- \asks_{u}\Bigr)\\
   \geq \ &  \Bigl(k_{s,j}-\epsilon\sigma_{s}-\theta k_{t,i} -(1-\theta)k_{u,l}  -\epsilon \Bigr) B(u)=0.
   \end{align*} 

   \ref{it:NACVB2}
    Next, assume that~\eqref{eq:NACVB2} fails. Then buying the contract
   \[
    C_u(K_{u,l})-CVB_s(\sigma_s, K_{s,j}) + k_{u,l}+\epsilon-(k_{s,j}-\epsilon\sigma_s)
   \]
   earns an initial profit. 
   If~$CVB_s(\sigma_s, K_{s,j})$ expires out of the money, then we leave the portfolio as it is. Otherwise we immediately enter a short position and close it 
   at time~$u$. The liquidation value is then non-negative,
   \begin{align*}
     (k_{s,j}-\epsilon\sigma_{s}+\epsilon)B(u)&+(\bids_s-S_s^C)\frac{B(u)}{B(s)}+ (S^C_{u}-K_{u,l})^+-\asks_u \\
                                              &+ \Bigl(k_{u,l}+\epsilon-(k_{s,j}-\epsilon\sigma_s) \Bigr)B(u) \geq 0.
   \end{align*} 

   \ref{it:NACVB3}
   If \eqref{eq:NACVB3} fails, then we buy the contract $C_t(K_{t,i})-CVB_s(\sigma_s,k_{s,j})$ for negative cost.
   Again we can focus on the case where $CVB_s(\sigma_s,k_{s,j})$ expires in the money.
   We sell one unit of the underlying at time~$s$ and close the short position at time~$t$. 
   The liquidation value of this strategy at time~$t$ is non-negative,
   \[
      (k_{s,j}-\epsilon\sigma_{s}+\epsilon)B(t)+ (\bids_s-S_s^C)\frac{B(t)}{B(s)}+ (S_t^C- K_{t,j})^+ - \asks_t \geq 0.
   \]
	
   \ref{it:NACVB4}
   We will show that there cannot exist an $\epsilon$-consistent model, if~\eqref{eq:NACVB4} fails. 
   In every model where the probability that $CVB_s(\sigma_s,k_{s,j})$ expires in the money is zero, 
   we could simply sell $CVB_s(\sigma_s,k_{s,j})$ and follow the trading strategy from
   Lemma~\ref{le:CVB}, realizing (model-dependent) arbitrage.
   On the other hand, if $CVB_s(\sigma_s,k_{s,j})$ expires in the money with positive probability, then we can use the same strategy as 
   in the proof of~\ref{it:NACVB3}. At time~$t$  the liquidation value of the portfolio is positive with positive probability.
\end{proof}

Note that, if $\epsilon=0,$ then $CVB_s(\sigma_s, k_{s,j})$ has the same payoff as $-C_s(K_{s,j})$. Keeping this in mind,
it is easy to verify that the conditions from Theorem~\ref{thm:NACVB} are a generalization of equations (4), (5) and (6) in~\cite{Co07}.

It remains open whether \eqref{eq:NACVB1}, \eqref{eq:NACVB2}, \eqref{eq:NACVB3} and \eqref{eq:NACVB4} are also sufficient for the existence of an $\epsilon$-consistent model. 

\begin{con}
Given the conditions stated in Theorem~\ref{thm:NACVB} the given prices are $\epsilon$-consistent with the absence of arbitrage if and only if 
\eqref{eq:NACVB1}, \eqref{eq:NACVB2}, \eqref{eq:NACVB3}, and \eqref{eq:NACVB4} hold.
There is weak arbitrage whenever \eqref{eq:NACVB1}, \eqref{eq:NACVB2}, and \eqref{eq:NACVB3} hold but~\eqref{eq:NACVB4} fails.
\end{con}

Theorem~\ref{thm:NACVB} can be used to find arbitrage opportunities associated with given market prices.
However, it might not be clear how to find parameters that satisfy the conditions
of Definition~\ref{def:CVB}. For the reader's convenience, we finish this section with an algorithm
which can be used to create CVBs given the prices at the beginning of Section~\ref{sec:notation}.
It is not hard to see that it yields all possible parameter configurations.
Once a particular CVB is chosen, its bid price can be obtained via~\eqref{eq:CVBprice}.

\begin{enumerate}[label=(\textit{\roman*})]
\item \label{it:algo1} Pick $j_1 \in \{0, \dots, N_1\}$ and $\sigma_1 \in \{-1,1\}$ and set $x_1= k_{1, j_1}-\epsilon\sigma_1$.
\item \label{it:algo2} Given $\{x_1,\dots, x_{t-1}\}$, $\{\sigma_1,\dots, \sigma_{t-1}\}$, 
                       $\{j_1, \dots, j_{t-1}\}$ and $\{i_2, \dots, i_{t-1}\}$ 
                       first pick $j_t \in \{0, \dots, N_{t}\}$. 
\item \label{it:algo3} Choose $\sigma_t$ distinguishing the following cases:
      \begin{itemize}
        \item if $k_{t, j_t} \geq x_{t-1}+\epsilon$ set $\sigma_t = -1$;
        \item if $k_{t, j_t} \leq x_{t-1}-\epsilon$ set $\sigma_t = 1$;
        \item if $k_{t, j_t} = x_{t-1}$ pick $\sigma_t \in \{-1,0,1\}$;
        \item if $k_{t, j_t} \in (x_{t-1}-\epsilon, x_{t-1}+\epsilon) \setminus \{x_{t-1}\}$ pick $\sigma_t \in \{-1,1\}$;
      \end{itemize}
\item \label{it:algo4} Set $x_t = k_{t, j_t}- \sigma_t\epsilon$ and pick $i_t \in \{0,\dots,N_t\}$ 
      such that either $k_{t,i_t} \leq x_{t-1}+\sigma_t\epsilon$ or $i_t=0$.
\item\label{it:algo5}  Repeat steps~\ref{it:algo2} to~\ref{it:algo4}.
\end{enumerate}



\section{Conclusion}\label{sec:conc}

We define the notion of $\epsilon$-consistent prices, meaning that a set of bid and ask prices for call options
and the underlying can be explained by a model with bid-ask spread bounded by~$\epsilon$.
For a single maturity, we solve the $\epsilon$-consistency problem, recovering the
trichotomy consistency/weak arbitrage/model-independent arbitrage from the frictionless
case~\cite{DaHo07}.
The interval of spread bounds for which a consistent model exists can be easily computed.
The multi-period problem seems to be rather difficult. As a first
step, we provide two results: Necessary explicit conditions, and equivalent semi-explicit
conditions. For the latter, we invoke a recent result from~\cite{GeGu18} on approximation
by peacocks. Finally, we note that Section~3.3 of the PhD thesis~\cite{Gu16}
discusses the multi-period problem under simplified assumptions. In particular,
it is assumed that only the underlying has a bid-ask spread, but not the options.

\appendix

\section{Proof of Theorem~\ref{thm:single}: $\epsilon$-consistency}\label{app:proof single}

   We first show that the conditions are necessary.
   Throughout the proof we will denote the option $C_1(K_{1,i})$ by $C^i$ to ease notation.

   \ref{it:NAsing1}
     Suppose that $1\leq i < j < l$ are such that~\eqref{eq:NAsing1} does not hold.
     We buy a  butterfly spread, which is the contract
     \[
       BF^{i,j,l}= \frac{1}{K_j-K_{i}}C^{i} + \frac{1}{K_{l}-K_{j}}C^{l}- \Bigl(\frac{1}{K_j-K_{i}}+\frac{1}{K_{l}-K_{j}} \Bigr)C^{j}
     \]
     and get an initial payment. Its payoff at maturity is positive if $S_1^C$ expires in the interval $(K_i, K_l)$ and zero otherwise, and so we have model-independent arbitrage. 
     
     If~\eqref{eq:NAsing1} fails for $i=0$ we buy the contract
     \[
       BF^{0,j,l}= \frac{1}{K_j-B \epsilon }S + \frac{1}{K_{l}-K_{j}}C^{l}- \Bigl(\frac{1}{K_j-B \epsilon}+\frac{1}{K_{l}-K_{j}} \Bigr)C^{j}
     \]
     and make an initial profit. 
     Note that $S$ denotes the underlying.
     At maturity the liquidation value of the contract is given by
     \[
       \frac{1}{K_j-B \epsilon }\bids_1+ \frac{1}{K_{l}-K_{j}}(S_1^C-K_l)^+ - \Bigl(\frac{1}{K_j-B \epsilon}+\frac{1}{K_{l}-K_{j}} \Bigr)(S_1^C-K_j)^+ 
     \]
     which is always non-negative.

   \ref{it:NAsing2}
     Suppose that~\eqref{eq:NAsing2} fails for $1 \leq i <l$. Then we buy a call spread $C^l-C^i$ and invest $k_l-k_i$ in the bank account. 
     This earns an initial profit, and at maturity the cashflow generated by the options is at least $K_i -K_l$, which means that we have arbitrage. 
     Now we consider the case where $i=0$. Note that in this case \eqref{eq:NAsing2} is equivalent to 
     \[
        \frac{\asky_{l}-\bids_0}{k_{l}+\epsilon} \geq -1. 
     \]
     If this fails we buy~$C^l$, sell one unit of the underlying, and invest $k_l+\epsilon$ in the bank account. 
     Again we earn an initial profit, and at maturity we close the short position and have thus constructed an arbitrage strategy.

   \ref{it:NAsing3}
      If \eqref{eq:NAsing3} fails for $0<i<j$, then we buy the call spread  $C^i-C^j$ and get an initial payment. Its payoff at maturity is always non-negative. 
      
      If \eqref{eq:NAsing3} fails for $i=0$, then we sell $C^j$ and buy one unit of the stock, which also yields model-independent arbitrage. 

   \ref{it:NAsing4}
      We show that we cannot find an arbitrage-free model for the given prices, if~\eqref{eq:NAsing4} fails. 
      Later, in Appendix~\ref{app:proof weak}, we will show that there is a weak arbitrage opportunity in this case 
      (which entails, according to Definition~\ref{def:wNA}, that there is no model-independent arbitrage).      
       
      In any model where $\pp(S_1^C > K_j)=0$ we could sell~$C^j$. As this option is never exercised, this yields arbitrage. 
      If on the other hand $\pp(S_1^C > K_j)>0$ and $i>0$, then we buy the call spread $C^i-C^j$ at zero cost. 
      At maturity the probability that the options generate a positive cashflow is positive. 
      If $i=0$, then we buy the contract $S-C^j$ instead, and at maturity the liquidation value of the portfolio 
      is given by $\bids_1-(S_1^C-K_j)$, which is positive with positive probability.
      This completes the proof of necessity.

    Now we show that the conditions in Theorem~\ref{thm:single} are sufficient for
    $\epsilon$-consistency, using Lemma~\ref{lem:overview}.
    We first argue that we may w.l.o.g.\ assume that $\asky_N=\bidy_N=0$.
    Indeed, we could choose 
   \[
     k_{N+1} \geq \max \Biggl\{ \frac{\asky_i  k_j - \bidy_j  k_i}{\asky_i-\bidy_j}: \ 0 \leq i<j \leq N,
      \ \asky_i-\bidy_j>0 \Biggr\} \vee \max\{k_j + \bidy_j: 0 \leq j \leq N\}
   \]
   and set $\asky_{N+1}=\bidy_{N+1}=0$. Then all conditions from Theorem~\ref{thm:single} would still hold, 
   if we included an additional option with strike $k_{N+1}$ and bid and ask price equal to zero. 
   So from now on we assume that $\asky_N=\bidy_N=0$.

   We will first show that, for $s \in \{0, \dots,N \}$, we can find $e_s \in [\bidy_s, \asky_s ]$ such that the linear interpolation $L$ of the points 
   $(k_s, e_s), s \in \{0, \dots,N \}$, is convex, decreasing, 
   and such that the right derivative of~$L$ satisfies $L'(k_0)\geq -1$. Then we will extend $L$ to a call function, and its associated measure will be the law of $D(1)S_1^C$.
   The sequence $(e_s)_{s \in \{1,\dots,N\}}$ can then be interpreted as 
   shadow prices of the options with strikes $(k_s)_{s \in \{1,\dots,N\}}.$

    Before we start we will introduce some notation. 
   For $j,l \in \{1,\dots,N\}, j<l$ we denote the line connecting $(k_j, \bidy_j)$ and $(k_l, \asky_l)$  by $f_{j,l}$, i.e.,
   \[
    f_{j,l}(x)= \bidy_{j}+ \frac{\asky_{l}-\bidy_{j}}{k_{l}-k_{j}} \cdot (x-k_{j}).
   \]
   If $e_s$ is known for some $s \in \{0 \dots,N \}$, then we denote the line connecting $(k_s,e_s)$ and $(k_i, \asky_i), i \in \{s+1, \dots, N\}$ 
   by $g_{s,i}$, i.e.,
   \[
    g_{s,i}(x)= e_s+ \frac{\asky_{i}-e_s}{k_{i}-k_{s}} \cdot (x-k_{s}).
   \]
   The linear interpolation of $(k_{s},e_{s})$ and $(k_j, \bidy_j), j \in \{s+1, \dots, N\} $ will be denoted by $h_{s,j}$,
   \[
    h_{s,j}(x)= e_s+ \frac{\bidy_{j}-e_s}{k_{j}-k_{s}} \cdot (x-k_{s}).
   \]
   We will refer to the slopes of these lines as $f'_{j,l},g'_{s,i}$ and $h'_{s,j}$ respectively. 

   First we will construct~$e_0$. In order to get all desired properties -- this will become clear towards the end of the proof 
   --~$e_0$ has to satisfy
   \begin{equation}\label{eq:e0low}
     e_0 \geq \max_{0 \leq j < l \leq N} f_{j,l}(k_0),
   \end{equation} 
   and 
   \begin{equation}\label{eq:e0up}
      e_0 \leq \min_{0\leq i \leq N } (k_i+ \asky_i-k_0 ).
   \end{equation} 
   We will argue that we can pick such an $e_0$ by showing that 
   \begin{align} \label{eq:e01}
      f_{j,l}(k_0) \leq k_i+ \asky_i-k_0, \quad  i,j,l \in \{0,\dots, N\}, j\leq l.
   \end{align}
   Using~\eqref{eq:NAsing2} twice we can immediately see that~\eqref{eq:e01} holds for $i \geq j$, 
   \[
    f_{j,l}(k_0) \leq  \bidy_j +k_j -k_0 \leq \asky_i+ k_i-k_0.
   \]
   If on the other hand $i <j$  we rewrite the right hand side of~\eqref{eq:e01} to $h_i(k_0)$, where $h_i(x)=-x+\asky_i+k_i$.
   Then from \eqref{eq:NAsing1} we get that
   \begin{equation*}
     f_{j,l}(k_i) \leq \asky_i= h_i(k_i),
   \end{equation*}
   and as $f'_{j,l} \geq -1 = h'_i,$ the inequality follows. 

   The above reasoning shows that existence of an $e_0$ such that~\eqref{eq:e0low} and~\eqref{eq:e0up} hold. 
   Next we want to construct~$e_1$ for given~$e_0$. It has to satisfy the requirements
   \begin{equation}\label{eq:e1low}
      e_1 \geq \max_{1\leq j<l \leq N} f_{j,l}(k_1)  \vee(e_0+k_0-k_1)
   \end{equation} 
   and 
   \begin{equation}\label{eq:e1up}
     e_1 \leq \min_{1 \leq i \leq N} g_{0,i}(k_1).
   \end{equation} 
   Again we will argue that we can pick such an $e_1$ by considering the corresponding inequalities. First note that the inequality
   \begin{equation*}
     e_0+k_0-k_1\leq g_{0,i}(k_1), \quad i \in \{1, \dots, N\},
   \end{equation*}
   follows directly from \eqref{eq:e0low}. Next we want to prove that
   \begin{equation}\label{eq:e11}
    f_{j,l}(k_1) \leq g_{0,i}(k_1), \quad i,j,l \in \{1,\dots,N\}, j <l.
   \end{equation}
   Therefore observe that 
   \begin{equation*}
    f_{j,l}(k_0) \leq e_0= g_{0,i}(k_0).
   \end{equation*}
   If $i < j$ \eqref{eq:e11} follows from~\eqref{eq:NAsing1}, because $f_{j,l}(k_i) \leq \asky_i= g_{0,i}(k_i)$. For $i=j$ 
    we may simply use the fact that $\bidy_i \leq \asky_i$ and hence we get that 
   $f_{j,l}(k_i) \leq \asky_i= g_{0,i}(k_i)$. For $i >j$ we may use  $f_{j,l}(k_0) \leq e_0 = h_{0,j}(k_0)$ to get
   \begin{equation*}
    f_{j,l}(k_1) \leq h_{0,j}(k_1) \leq g_{0,i}(k_1), 
   \end{equation*}
   where the last inequality follows from the fact that $h_{0,j}(k_0) = g_{0,i}(k_0)=e_0$ and that 
   \begin{equation*}
    h'_{0,j} = \frac{\bidy_j-e_0}{k_j-k_0} \leq \frac{\asky_i-e_0}{k_i-k_0} = g'_{0,i}. 
   \end{equation*}
   In the last step we used that $e_0 \geq f_{j,i}(k_0)$. 

   Now suppose we have already constructed $e_{1}, \dots e_{s-1}, s \in {1,\dots, N}$. Then for $r \in \{1, \dots s-1 \}$ we have that
   \begin{align}\label{eq:erlow}
     e_r \geq \biggl(e_{r-1}+\frac{e_{r-1}-e_{r-2}}{k_{r-1}-k_{r-2}} \cdot (k_r-k_{r-1})\biggr)  \vee \max_{ r \leq j<l \leq N }  f_{j,l}(k_r),
   \end{align}
   and 
   \begin{align}\label{eq:erup}
   e_r \leq \min_{r \leq i \leq N} g_{r-1,i}(k_r).
   \end{align}
   Note that for $r=1$ we need an appropriate~ $e_{-1}$ and $k_{-1}$ in order for~\eqref{eq:erlow} to hold.
   For instance, we can set
   $k_{-1}=-1$ and $e_{-1}= e_0-(k_{0}+1)\cdot (e_1-e_0)/(k_{1}-k_{0})$.

   We want to show that we can choose $e_{s}$ such that \eqref{eq:erlow} and \eqref{eq:erup} hold for $r=s$.
   First, the inequality
   \[
    e_{s-1}+\frac{e_{s-1}-e_{s-2}}{k_{s-1}-k_{s-2}} \cdot (k_s-k_{s-1}) \leq g_{s-1,i}(k_s),
    \quad i \in\{s, \dots, N\},
   \]
   is equivalent to 
   \[
    \frac{e_{s-1}-e_{s-2}}{k_{s-1}-k_{s-2}} \leq \frac{\asky_i-e_{s-1}}{k_{i}-k_{s-1}}
   \]
   which is again equivalent to
   \[
    e_{s-1} \leq g_{s-2,i}(k_{s-1})
   \]
   and holds by \eqref{eq:erup}.

   The inequality
   \[
    f_{j,l}(k_{s})  \leq g_{s-1,i}(k_s), \quad i,j,l \in \{s,\dots,N\}, j<l,
   \]
   can be shown using the same arguments as before:
   first we note that $f_{j,l}(k_{s-1})  \leq e_{s-1}= g_{s-1,i}(k_s)$ and then we distinguish between $i<j$, $i=j$ and $i>j$.

   We have now constructed a finite sequence $(e_s)_{s \in \{0, \dots,N \}}$. 
   Observe that for all $s \in \{0, \dots,N \}$ the bounds on $e_s$ from above, namely~\eqref{eq:e0low} and~\eqref{eq:e0up} for $s=0$,~
   \eqref{eq:e1low} and~\eqref{eq:e1up} for $s=1$ and~\eqref{eq:erlow} and~\eqref{eq:erup} for $s>1$, ensure that $e_s \in [\bidy_s,\asky_s]$. 
   Denote by $L: [k_0, k_N]\rightarrow \setR$ the linear interpolation of the points $(k_s,e_s), s \in \{0, \dots,N \}$. Then $L$ is convex,
   which is easily seen from
   \[
     e_s \geq e_{s-1}+\frac{e_{s-1}-e_{s-2}}{k_{s-1}-k_{s-2}} \cdot (k_s-k_{s-1}), \quad s \geq 2.
   \]
   Furthermore, by \eqref{eq:e1low}
   \[
    L'(k_0)= \frac {e_1-e_0}{k_1-k_0} \geq -1.
   \]
   Finally, $L$ is strictly decreasing on $\{L>0\}$ which is most easily seen from $e_{s} \leq g_{s-1,N}(k_s)$.
   Therefore~$L$ can be extended to a call function~$R$ as follows (see Proposition~2.3
   in~\cite{GeGu18}),
   \begin{align*}
     R(x)=
		\begin{cases}
		   L(k_0)+k_0-x, & x \leq k_0,\\
		   L(x),& x \in [k_0,k_N],\\
		   0,& x\geq k_N.
		\end{cases}
   \end{align*}
   Let $\mu$ be the associated measure. Then $\mathbb{E}\mu = R(0) =  L(k_0)+k_0 \in [\bids_0-\epsilon,\asks_0+\epsilon]$.
   If $\mathbb{E}\mu < \bids_0$ we define a measure $\nu$ by setting $\nu(A)=\mu(A-\epsilon)$ for Borel sets~$A$.
   The set $A-\epsilon$ is defined as $\{a-\epsilon: a \in A \}$. Then $\mathbb{E}\nu=\mathbb{E}\mu+\epsilon \in [\bids_0,\asks_0]$.
   Similarly, if $\mathbb{E}\mu > \asks_0$ we define $\nu(A)=\mu(A+\epsilon)$ for Borel sets~$A$,
   and if $\mathbb{E}\mu \in [\bids_0,\asks_0]$ then we simply set $\nu=\mu$. 
   Furthermore for $x < k_0$ we have that $R'(x)=-1$, therefore $\mu$ has support $[\epsilon,\infty)$.
   Clearly, by definition of~$\nu$, we have that $W^\infty(\mu,\nu)\leq \epsilon$.
   Hence, by Lemma~\ref{lem:overview} the prices are~$\epsilon$-consistent with the absence of arbitrage.

\section{Proof of Theorem~\ref{thm:single}: weak arbitrage}\label{app:proof weak}

   As we have seen in part~\ref{it:NAsing4} of the necessity proof of 
   Theorem~\ref{thm:single} (see Appendix~\ref{app:proof single}),
   there is an arbitrage opportunity that depends on
   the null sets of the model. We will show that there is no model-independent arbitrage strategy. 
   Suppose, on the contrary, that there is one.
   Then we can construct a portfolio $\phi_1^0 +\phi_1^1 S +\sum_{l=1}^N \phi^l C(K_l)$, where $\phi_1^0, \phi_1^1, \phi^l \in \setR$,
   such that its initial cost is negative, i.e.,
   \[
      \phi_{1}^0 + \bigl((\phi_1^{1})^+\asks_0- (\phi_1^{1})^-\bids_0 \bigr)+ \sum_{l=1}^N \bigl( (\phi^l)^+\asky_l- (\phi^l)^-\bidy_l \bigr)<0,
   \]
   and such that the liquidation value at maturity is non-negative, i.e.,
   \[ 
    \phi_{1}^0 B(1)+\bigl((\phi_1^{1})^+\asks_1- (\phi_1^{1})^-\bids_1 \bigr)+\sum_{l=1}^N \phi^l (S_1^C-K_l)^+ \geq 0.
   \]
   Without loss of generality we can assume that $|\phi_1^0|+ |\phi_1^1|+\sum_{l=1}^N |\phi^l|=1$. 

   Next we construct $e_0, \dots, e_N$ as in the sufficiency proof of Theorem~\ref{thm:single}.
   Clearly, we then have $\asky_i=e_i=e_{i+1}=\dots=e_N$. 
	 The idea is to consider a market with slightly different shadow prices $\widetilde{e}_l$,
	 which can be obtained from the original shadow prices $e_l$ by shifting them down.
	 More precisely, we set
	 \[
	   l_0 = \max\{ l : 0\leq l\leq N,\ e_l+k_l=e_0+k_0 \},
	 \]
	 define
   \[
      z= \min \Biggl\{-\frac {r_\Phi}{2},\
		  \bigl(e_{l_0+1}+k_{l_0+1}-e_{l_0}-k_{l_0}\bigr)\cdot\frac{\sum_{s=l_0}^N \limits (k_s-k_{l_0})}{k_{l_0+1}-k_{l_0}},
		  \ e_N \cdot \frac{\sum_{s=l_0}^N \limits (k_s-k_{l_0})}{k_N-k_{l_0}}\Biggr\},
   \]
   and put $\widetilde{e}_l =e_l$ for $l \leq l_0$ and for $l> l_0$ 
   \[
    \widetilde{e}_l = e_l - z \frac{k_l-k_{l_0}}{\sum_{s=l_0}^N \limits (k_s-k_{l_0})}.
   \]
   Now consider a modified set of prices, where bid and ask price of the $l$-th call, $0\leq l \leq N$,
   are both defined by~$\widetilde{e}_l$.
   It is easy to check that these prices satisfy
   all conditions from Theorem~\ref{thm:single},
   and hence do not admit any arbitrage opportunities.
   Indeed, the second expression  in the definition of $z$ guarantees 
   that $e_{l_0+1}$ is not too small, i.e.,
   \[
     \frac{e_{l_0+1}-e_{l_0}}{k_{l_0+1}-k_{l_0}}\geq-1,
   \]
   and the third expression ensures that $\widetilde{e}_N$ is not too small, i.e., $\widetilde{e}_N \geq 0$.
   A simple calculation shows that
   \begin{align*}
      \phi_{1}^0 + \bigl((\phi_1^{1})^+\asks_0- (\phi_1^{1})^-\bids_0 \bigr) +\sum_{l=1}^N \phi^l \widetilde{e}_l &
    = \phi_{1}^0 + \bigl((\phi_1^{1})^+\asks_0- (\phi_1^{1})^-\bids_0 \bigr) +\sum_{l=1}^N \phi^l e_l - \sum_{l=l_0+1}^N \phi^l (e_l-\widetilde{e}_l) \\ 
    & \leq r_\Phi- \sum_{l=l_0+1}^N \phi^l (e_l-\widetilde{e}_l) \\
		& \leq r_\Phi+ z \sum_{l=l_0+1}^N |\phi^l| \frac{k_l-k_{l_0}}{\sum_{s=l_0}^N \limits (k_s-k_{l_0})}  \\
		& \leq r_\Phi+ z  \leq \frac {r_\Phi}{2} < 0,
   \end{align*}
   and so the portfolio $C_{\Phi}$ in the modified market has negative cost. 
   But its liquidation value at maturity is unchanged and hence non-negative, 
   and we have thus constructed a model-independent arbitrage strategy for the modified set of prices, which is a contradiction.

\bibliographystyle{siam}
\bibliography{literatur}

\end{document}